\pdfoutput=1

\documentclass[conference]{IEEEtran}
\usepackage{cite}

\ifCLASSINFOpdf
\else
\fi
\usepackage{amsmath}
\usepackage{amssymb,amsthm}
\DeclareMathSizes{10}{9}{8}{7}
\newtheorem{lemma}{Lemma}
\newtheorem*{remark}{Remark}

\usepackage[capitalise]{cleveref}
\crefname{section}{§}{§§}
\Crefname{section}{§}{§§}
\crefformat{section}{§#2#1#3}
\crefname{paragraph}{§}{§§}
\Crefname{paragraph}{§}{§§}
\crefformat{paragraph}{§#2#1#3}

\crefname{appendix}{Appendix}{Appendices}
\Crefname{appendix}{Appendix}{Appendices}
\crefname{figure}{Figure}{Figures}
\Crefname{figure}{Figure}{Figures}
\crefname{equation}{Equation}{Equations}
\Crefname{equation}{Equation}{Equations}
\crefname{table}{Table}{Tables}
\Crefname{table}{Table}{Tables}
\crefname{lemma}{Lemma}{Lemmas}
\Crefname{lemma}{Lemma}{Lemmas}
\crefname{theorem}{Theorem}{Theorems}
\Crefname{theorem}{Theorem}{Theorems}
\crefname{problem}{Problem}{Problems}
\Crefname{problem}{Problem}{Problems}
\crefname{algorithm}{Algorithm}{Algorithms}
\Crefname{algorithm}{Algorithm}{Algorithms}
\crefname{algocf}{Algorithm}{Algorithms}
\Crefname{algocf}{Algorithm}{Algorithms}
\crefname{alg}{Algorithm}{Algorithms}
\Crefname{alg}{Algorithm}{Algorithms}

\usepackage[shortlabels]{enumitem}
\setlist[enumerate]{1.,topsep=0pt,itemsep=-1ex,partopsep=1ex,parsep=1ex,leftmargin=15pt}
\usepackage{tikz}
\usepackage{pgfplots}
\usepackage{symbols}
\usepackage{xspace}
\usepackage[font+=small,labelfont+=bf]{subcaption}
\usepackage{listings}
\usepackage{flushend}
\usepackage{booktabs}
\usepackage{url}

\definecolor{vintageblack}{HTML}{484043}
\definecolor{vintageyellow}{HTML}{FFC805}
\definecolor{vintagegreen}{HTML}{38A528}
\definecolor{vintageorange}{HTML}{FF4D25}
\definecolor{vintagepurple}{HTML}{AE56E2}

\newcommand{\system}{\textsc{Airphant}\xspace}

\newcommand{\tofix}[1]{{\fcolorbox{red}{pink}{\parbox{8cm}{#1}}}}
\newcommand{\fixed}[1]{{\color{purple} #1}}

\begin{document}
%
\title{\system: Cloud-oriented Document Indexing}


\author{\IEEEauthorblockN{
Supawit Chockchowwat,
Chaitanya Sood and
Yongjoo Park
}
\IEEEauthorblockA{
University of Illinois at Urbana-Champaign\\
\{supawit2,csood2,yongjoo\}@illinois.edu}}


%


\maketitle

\begin{abstract}
Modern data warehouses
can
scale compute nodes independently of storage. These systems persist their data on cloud storage, which is always available and cost-efficient. Ad-hoc compute nodes then fetch necessary data on-demand from cloud storage.
This ability to quickly scale or shrink data systems is highly beneficial
if query workloads may change over time.

We apply this new architecture to \emph{search engines}
with a focus on optimizing their latencies in cloud environments.
However, simply placing existing search engines (e.g., Apache Lucene) on top of cloud storage significantly
increases their end-to-end query latencies (i.e., more than 6 seconds on average in one of our studies).
This is because their indexes can incur multiple network round-trips
due to their hierarchical structure (e.g., skip lists, B-trees, learned indexes).
To address this issue,
we develop a new statistical index (called IoU Sketch). For lookup, IoU Sketch makes multiple asynchronous network requests in parallel.
While IoU Sketch may fetch more bytes than existing indexes, it significantly reduces the index lookup time because parallel requests do not block each other.
Based on IoU Sketch, we build an end-to-end search engine, called \system;
we describe how \system builds, optimizes, and manages IoU Sketch; and ultimately, supports keyword-based querying.
In our experiments with four real datasets,
\system's average end-to-end latencies are between 13 milliseconds and 300 milliseconds, 
being up to 8.97$\times$ faster than Apache Lucence and 113.39$\times$ faster than Elasticsearch.
\end{abstract}


%
\IEEEpeerreviewmaketitle

\section{Introduction}
\label{sec:intro}


In public clouds, 
compute and storage can be scaled independently.
Based on this,
modern data warehouses including Snowflake~\cite{dageville2016snowflake}, Google BigQuery~\cite{google-bigquery}, Amazon Redshift~\cite{amazon-redshift}, and Presto~\cite{prestodb,prestosql}
allow users to quickly
start new clusters,
analyze a large volume of data retrieved on demand from cloud storage (e.g., AWS S3~\cite{aws-s3}, Azure Blob Storage~\cite{azure-blob}, GCP Cloud Storage~\cite{gcp-storage}),
and finally terminate those clusters if they are no longer needed.\footnote{Google BigQuery effectively takes this series of operations. However, they are hidden from users' perspective; the users are charged based on the cost of each query.}
Users can easily tune many properties of those clusters such as number of nodes, number of cores, memory size.
To maximize the performance, 
users can upgrade their clusters.
To reduce the cost, users can downsize their clusters.
Users can thus fine-tune the overall cluster bill
and expected performance; 
regardless, their data is kept independently in cloud storage, offering cost-efficiency and robustness.\footnote{Note that compute resources (e.g., cores, memory) typically takes a larger portion in cloud bills than storage. 
Thus, scaling down compute nodes can lead to bigger savings.}
This capability to quickly scale up and down clusters has been well received.\footnote{For example, as of early 2021, Snowflake is one of the fastest growing private companies in the data space, with its market capitalization reaching about \$70 billion~\cite{cnn}. Also, Amazon Redshift introduced in December 2019 a new type of node called \texttt{ra3},
which offers the flexibility in adjusting cluster sizes~\cite{redshift-ra3}.}
We refer to this new architectural shift as \emph{separation of compute and storage}.

\paragraph{New Direction}

We investigate applying this new architectural shift to
\emph{search engines}. Search engines are systems for retrieving relevant documents that contain user-provided search keywords.
Commonly used search engines all rely on some form of indexes (e.g., skip lists~\cite{pugh1990skip}, B+ tree~\cite{Bayer1972,Chen2002,Chen2001}, learned indexes~\cite{Kraska2018,Ding2020}) to quickly find relevant documents.
For fast index traversals, 
they need their compute and storage closely located on the same machine.
This architecture has a disadvantage that we need to keep a large cluster running as more documents are indexed,
even if query workloads are light or even if some documents are queried infrequently.

However, if we can build a special index stored entirely in cloud storage,
we can scale compute nodes independently of indexing. This offers flexibility when query workloads change over time or inquire about particular documents more frequently than others. For instance, rarely queried documents can simply be dormant in low-cost cloud storage along with their index.
This new approach, if possible, can enable highly cost-efficient
searching.

\paragraph{Challenges}

Existing search engines do not perform well under the separation of compute and storage.
Because existing indexes have hierarchical structures~\cite{pugh1990skip,Bayer1972,Chen2002,Chen2001,Kraska2018,Ding2020,bialecki2012apache}, they incur significant traversal overhead when
the data is moved further away from compute.
To identify \emph{postings} (i.e. references to documents) of documents that contain the search keywords, or specifically, to locate the data block containing those postings (called a \emph{postings list}), an existing index requires traversing from its root, to a child node at the next level, and so on, until we reach a leaf node. In this process, to go one level deeper, we need to fetch the next node from the storage device, incurring an additional communication.
This means that to reach a leaf node at level $N$, we need to make $N$ sequential back-to-back communications.
If those individual communications are fast enough (e.g., local SSDs), the end-to-end latency for $N$ communications may be tolerable. 
However, if those nodes are stored in cloud storage, each communication needs a network round-trip, 
which can be orders of magnitude slower than local SSDs. If such communications must be made sequentially multiple times,
the total end-to-end latency quickly adds up, significantly affecting end-user experiences.
Although node caching may reduce communications, allocating a large enough cache to store the entire index is prohibitively expensive when the corpus size is large. For this reason, it is generally hard to avoid multiple sequential communications in existing indexes.
On a separate note, Elasticsearch has a S3 plugin, only for snapshots~\cite{elasticsearch-plugin}, but not for interactive querying.

\paragraph{Our System}

We introduce \system,\footnote{A portmanteau from ``air'' and ``elephant'', signifying the lightweight system capable of serving from a large corpus} \emph{a new search engine specifically designed for low-latency querying under the separation of compute and storage}.
While \system requires several nontraditional design decisions, its core idea is straightforward:
\emph{to minimize end-to-end query latencies, we should completely avoid sequential back-to-back communications with cloud storage.}
Instead, \system issues multiple asynchronous requests in parallel (thus, no blocking in between)
to obtain a \emph{postings list}.
This asynchronous approach is enabled by our novel statistical index.
Once we have a postings list, the rest of the process is almost identical; 
\system retrieves individual documents and presents them to users.

As noted above, the crux of \system is our statistical index
called \emph{Intersection of Unions Sketch (IoU Sketch)} (\cref{sec:indexing}).
Unlike existing index structures,
IoU Sketch initially produces multiple \emph{super postings lists (superpost)},
where each superpost contains both: 1) all relevant postings of documents containing search keywords,
and 2) some irrelevant postings of the documents not containing search keywords.
These superposts have two important properties.
First, they are independent from each other in retrieval process; thus, we can retrieve them in parallel.
Second, by intersecting them, the number of irrelevant documents reduces exponentially due to their incoherence enforced by randomization.
Thus, the final postings list---the intersection of all superposts---is \emph{almost} identical 
to the one from existing indexes.

Yet the final postings list may contain a few false positives in expectation. 
Nonetheless, \system removes them as it retrieves actual content of documents at a later step, recovering a perfect precision as a result.
Also importantly, superposts and so the final postings list contain no false negative; \system recalls all documents that indeed contain search keywords.
Moreover, IoU Sketch is carefully optimized according to a word frequency distribution in the corpus and queries (\cref{sec:indexing}) given memory and accuracy constraints.
Although these false positives lead to an additional cost in document retrieval,
IoU Sketch's performance benefit overcomes such cost, leading to an improvement over existing indexing techniques.
As a result, \system significantly outperforms existing search engines deployed similarly, keeping its query latencies always under a second even for the largest dataset we have tested.
IoU Sketch is built per corpus; thus, it is sufficient to download it only once before querying. In addition, its size is configurable (less than 2MB in our experiments), making IoU Sketch lightweight to keep it in memory.

\paragraph{Summary of Contributions}
The contributions of this work are organized as follows:
\begin{enumerate}
\item We introduce \system, a new search engine optimized for cloud environments.
Following the separation of compute and storage, \system persists all data (indexes and documents) in cloud storage.
(\cref{sec:motivation} and \cref{sec:design})
\item We describe a novel statistical index called IoU Sketch
as a core component of \system.
IoU Sketch can accurately identify documents containing search keywords with a single batch of concurrent communications with cloud storage; consequently, its end-to-end latency is significantly shorter than existing indexes.
(\cref{sec:airphant})
\item We show how we can optimize IoU Sketch for a given corpus and requirement. Specifically, we express IoU Sketch's accuracy in terms of expected number of false positives
parameterized by IoU Sketch structure. Then, we analyze and present an efficient optimization algorithm. (\cref{sec:airphant})
\item We empirically compare and analyze the performance of \system to other existing systems such as Lucene~\cite{lucene}, Elasticsearch~\cite{elasticsearch}, SQLite~\cite{sqlite}, and na\"ive hash table. (\cref{sec:exp})
\end{enumerate}

\section{\system: Core Ideas}
\label{sec:background}
\label{sec:motivation}


This section introduces the core ideas of \system.
For this, we first briefly recap core concepts in search engines (\cref{sec:motivation:background}).
Next, we discuss system challenges of running search engines on top of cloud storage (\cref{sec:motivation:challenges}).
Finally, we present the core ideas of \system for addressing those challenges (\cref{sec:motivation:approximate:index}).

\subsection{Search Engines: Background}
\label{sec:motivation:background}

This section overviews the core concepts in document indexing and searching~\cite{schutze2008introduction}.
First, we describe high-level user interface of search engines.
Next, we describe how a search engine quickly finds the documents containing a search keyword.

\paragraph{User Interface}

We briefly describe a typical programming interface of search engines from end users' perspectives.
\cref{fig:code:search_engine} shows an example code snippet for creating/indexing documents and performing a keyword-based search. The strings passed to document objects (i.e., \texttt{"hello world"} and \texttt{"hello airphant"}) are parsed by the search engine into multiple words (i.e., \texttt{"hello"} and \texttt{"world"} for doc1; \texttt{"hello"} and \texttt{"airphant"} for doc2).
These parsing rules are configurable by users~\cite{lucene-codecs}.
Then, the user can use any of those parsed keywords (i.e., hello, world, airphant) for searching.

\paragraph{Internal Workflow}

A typical search engine internally implements an inverted index to retrieve relevant documents.
An inverted index is a data structure to quickly identify the documents containing a search keyword.
It normally consists of two sub-components: 
1) \emph{postings lists}, and 2) \emph{term index}.
A postings list is a list of document IDs (e.g., doc1 and doc2) containing an associated keyword (e.g., ``hello'').
There are as many postings lists as the number of indexed keywords.
The term index is a map from each keyword (e.g., ``hello'') to the location of its associated postings list.
For fast searching, the term index should be able to quickly return the location of the postings list.
B-trees and skip lists and commonly used data structures for their asymptotically fast lookup operations.
Inverted indexes (both postings lists and term index) are constructed when documents are inserted into a search engine.

Given a query keyword, a search engine performs the following operations.
First, using the term index, it identifies the location of the postings list associated with the search keyword.
Second, it fetches the postings list for the search keyword.
Third, it fetches the actual contents of the documents and return them to the user.

For the term index, hierarchical indexes are commonly used.
A skip list~\cite{pugh1990skip,kirschenhofer1995analysis}
is used by Apache Lucene~\cite{lucene}, an underlying engine for distributed search engines such as ElasticSearch~\cite{elasticsearch} and Apache Solr~\cite{smiley2015apache}.
While it is probabilistic, a skip list needs $O(\log n)$ steps on average for each lookup.
A B-tree can also be used for the term index, which has the same lookup cost $O(\log n)$.


\begin{figure}[t]
\begin{lstlisting}[
    language=java,
    basicstyle=\ttfamily\footnotesize]
// index two documents
Index index = new Index();
Document doc1 = new Document("hello world");
index.addDocument(doc1);
Document doc2 = new Document("hello airphant");
index.addDocument(doc2);

// search for documents containing "airphant"
Document[] docs = index.search("airphant");
\end{lstlisting}
\vspace{-3mm}
\caption{Example user interface of a search engine}
\label{fig:code:search_engine}
\vspace{-2mm}
\end{figure}

\subsection{Search Engines on Cloud: Challenges}
\label{sec:motivation:cloud:search}
\label{sec:motivation:challenges}

Placing search engines directly on top of cloud storage increases their end-to-end search latencies significantly.
One critical performance bottleneck comes from term indexes.
We first explain why a term index lookup operations slow down substantially.
Then, we describe why simple alternative approaches are not likely to solve the issue.
These challenges motivate \system (\cref{sec:motivation:approximate:index}).

\paragraph{Cloud Storage} 

Major public cloud computing vendors (e.g., AWS, GCP, Azure) all offer cloud storage services. Simply speaking, these services are object storage where each object or \emph{blob}
is identified by a name. To upload and download data to and from cloud storage, vendor-provided programming APIs are used. These APIs make requests over the network to fetch data. These services are robust against data loss (e.g., 99.999999999\% durability of AWS S3) and cost-efficient (\$20/TB/month at the time of writing).

\paragraph{Performance Challenge}

While the network bandwidth is increasing,
we observe that \emph{network latency}---the time to get the first byte---acts 
as a critical bottleneck for search engines.
That is, whether we are retrieving a smaller or larger volume of data, there is some overhead we must pay for
every network request.
This network latency causes an affine relationship between 1) the size of data to fetch and 2) the total elapsed time
(referred to as \emph{retrieval latency}).
\cref{fig:network} depicts such a relationship using a small virtual machine (2 cores, 2GB memory) with multithreaded download on Google Cloud
Note that both X-axis and Y-axis are in log-scale.
The figure shows that the retrieval latency remains around 50 milliseconds until we increase the data size beyond 2MB. After that point, the retrieval latency increases linearly.

This affine characteristic of network performance causes existing search engines to perform poorly.
To find the postings list for a keyword, a search engine performs a lookup using a term index.
This lookup may involve multiple round-trips to storage devices for fetching intermediate nodes.\footnote{See our description on term index operations in \cref{sec:motivation:background}.}
If the storage device is cloud storage (not a local SSD), each round-trip incurs a network latency. 
When multiple network requests are made, lookup latency quickly increases even if each trip retrieves a small amount of data.
Parallelizing these network requests is non-trivial because in order to make the \emph{next} network request,
we need the information obtained from the \emph{previous} network request.
Of course, if we can cache the entire term indexes in memory for every corpus, we can completely avoid
any communications over the network. However, such scenario requires expensive equipment and so the approach is costly.
As described in \cref{sec:intro}, our goal is the opposite: we want to minimize our compute resources without much sacrificing search performance.

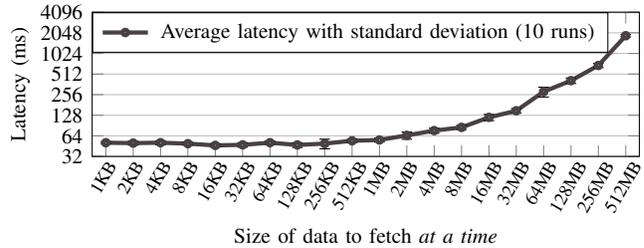
\begin{figure}[t]
\hspace*{-2mm}
\begin{tikzpicture}
    \begin{axis}[
        width=1.0\linewidth,
        height=35mm,
        xmin=9.5,
        xmax=29.5,
        xtick={10, 11, ..., 29},
        xticklabels={1KB,2KB,4KB,8KB,16KB,32KB,64KB,128KB,256KB,512KB,
                1MB,2MB,4MB,8MB,16MB,32MB,64MB,128MB,256MB,512MB},
        xticklabel style={rotate=60,xshift=1.0mm,yshift=1mm},
        xlabel near ticks,
        ylabel near ticks,
        ylabel style={align=center,yshift=-1mm},
        ymin=32,
        ymax=4096,
        ytick={1, 2, 4, 8, 16, 32, 64, 128, 256, 512, 1024, 2048, 4096},
        yticklabels={1, 2, 4, 8, 16, 32, 64, 128, 256, 512, 1024, 2048, 4096},
        xticklabel style={yshift=0mm},
        xlabel style={yshift=1mm},
        legend style={
            at={(0.0,1.0)},anchor=north west,column sep=2pt,
            draw=black,fill=white,
            /tikz/every even column/.append style={column sep=5pt}
        },
        legend cell align={left},
        legend columns=1,
        clip=false,
        every axis/.append style={font=\footnotesize},
        every x tick label/.append style={font=\scriptsize},
        ymode=log,
        minor grid style=lightgray,
        ymajorgrids,
        xlabel=Size of data to fetch \emph{at a time},
        ylabel=Latency (ms)]
    
    \addplot[
        ultra thick,draw=vintageblack,mark=*,mark options={scale=0.5, fill=white},
        error bars/.cd, y dir=both,y explicit,error bar style={color=black}]   
    table[x expr=\thisrowno{0},y=y,y error expr=\thisrowno{3}] {
    x c y    y+
    10	3	50.749	1.639488064
    11	3	50.116	3.264663331
    12	3	50.825	3.412556064
    13	3	49.164	2.196857351
    14	3	46.399	2.638292588
    15	3	47.259	3.531060432
    16	3	50.85	2.852203826
    17	3	47.063	3.744673402
    18	3	49.307	7.954900726
    19	3	54.194	2.950563788
    20	3	55.68	3.645176417
    21	3	64.982	7.988933735
    22	3	76.311	6.356670075
    23	3	85.126	5.952012171
    24	3	119.071	12.46208423
    25	3	148.809	11.50798414
    26	3	282.746	46.46524632
    27	3	410.031	37.38845379
    28	3	683.203	46.70860141
    29	3	1858.085	74.6925704
    };
    
    \addlegendentry{Average latency with standard deviation (10 runs)}
  
      
    \end{axis}
    \end{tikzpicture}

\vspace{-2mm}
\caption{End-to-end latency between Google Cloud virtual machines and Google Cloud Storage.
}
\label{fig:network}
\vspace{-2mm}
\end{figure}

\subsection{Statistical Approach to Indexing}
\label{sec:motivation:approximate:index}

In this section, we present the core ideas of our statistical indexing.
Its primary goal is to avoid multiple sequential communications in obtaining a postings list.
To achieve this, we propose non-trivial changes to inverted indexes.
In essence, it makes the following systems tradeoff:
the payloads of individual requests become larger, but those requests are made in parallel.

\paragraph{Bins}

To keep the pointers to \emph{all} postings lists in memory,
we form \emph{bins} by merging multiple keywords. Accordingly, associated postings lists are also merged into one. 
For example, consider three regular keywords and their postings lists: (``hello'' $\rightarrow$ (doc1, doc2)), (``world'' $\rightarrow$ (doc1)), and (``airphant'' $\rightarrow$ (doc2, doc3)). 
Suppose we merge ``hello'' and ``world'' into one bin $b_1$, which produces:
($b_1$ $\rightarrow$ (doc1, doc2)), (``airphant'' $\rightarrow$ (doc2, doc3)).
After merge, the number of keywords is reduced from three to two, but we no longer have exact postings list for merged keywords.
In other words, the postings list for $b_1$ contain both the postings list for ``hello'' and the list for ``world''. Thus, some of the documents in (doc1, doc2) may not contain ``hello'' (false positives); however, no other documents contain ``hello'' (\emph{no false negatives}).
Systematically, we rely on hash functions to determine how to group original keywords into bins.


This merge operation has one drawback: we will have to fetch more documents.
This drawback is more significant if there are much fewer bins than original keywords (e.g., 1000 keywords to 1 bin on average). 
Note that although the false positive documents can be removed by examining the actual content, in order to do so,
we must first fetch those documents.
Fetching, say, $1000\times$ more documents can slow down end-to-end search latencies greatly.
To address this, we exploit the probabilistic nature of random merge, as described below.

\paragraph{Multi-layer Structure}

Retrieving a large number of actual documents is slow. To avoid this, we reduce false positives by extending the above approach.
Note that merging keywords produces a map from a bin to a (merged) postings list. 
By repeating this merge operation, we construct a multi-layer (say $L$ layers) structure where each layer contains a different hash function.
Different hash functions produce different groupings;
thus, each layer is associated with different sets of bins. 

Given a keyword, we use this $L$-layer structure to obtain an accurate postings list, as follows. First, we apply $L$ hash functions to a keyword, and look up those hash values in our multi-layer maps; then, we obtain the \emph{pointers} to $L$ postings lists. 
Second, we fetch $L$ postings lists in parallel from cloud storage (here, we make a single batch of concurrent network requests). 
Finally, we intersect those $L$ postings lists to obtain a final posting list.

The final postings list is accurate (i.e., small false positives) because most false positives are eliminated
as part of intersection operations.
This is due to \emph{the multiplication rule for independent events} in probability.
That is, if we consider a single merge operation, it introduces many false positives into a postings list (from the perspective of one of the merged keywords).
However, if we consider multiple independent merge operations and take only the common postings among them, the number of false positives decreases exponentially.
We formally discuss its mathematical properties in \cref{sec:airphant}.
The multi-layer structure is still memory-efficient because it only stores hash functions and the pointers to (merged) postings lists; it does not store original keywords.




\section{System Design}
\label{sec:design}

We describe \system's core building blocks and concepts,
with a focus on their high-level operations.
Specifically, we present the current scope of \system (\cref{sec:scope}), its internal components, and offline operations (e.g., indexing building).




\subsection{Current Scope}
\label{sec:scope}

\system supports important use cases, but its capability is still limited compared to
full-fledged search engines such as ElasticSearch~\cite{elasticsearch} and Apache Solr~\cite{smiley2015apache}.
This section describes what \system currently supports: (a) the query workloads it can handle and (b) a few requirements it needs from cloud storage.

\paragraph{Documents and Queries}

\system offers document searches with exact keyword matching; that is, it finds the documents containing a given set of keywords. 
In comparison, existing full-fledged search engines also support others such as range queries, 
fuzzy queries, etc.
\system aims towards read-oriented workloads where the corpus doesn't change frequently. This focus allows many design decisions to fully harness strength of cloud storage. As of now, extensions to support a wider class of queries and handle frequent corpus updates are deferred to future works.

\paragraph{Target Environments}

\system targets modern public cloud,
where compute nodes can quickly turn on and off while most data can be stored safely in cloud storage.
\system heavily relies on cloud storage for data management. It indexes the documents stored in cloud storage and subsequently persists index structures in cloud storage as well. It assumes that cloud storage offers random read operations; that is, fetching bytes from an arbitrary offset doesn't require full read. 
Random reads are useful in packing multiple postings lists in a single blob; \system can read arbitrary postings lists without performance penalty. Also, keeping all postings lists in a few blob makes the overall data management easier because otherwise, we will need as many blobs as the number of bins.
Random reads are already supported by major cloud vendors~\cite{gcp-random-read,aws-s3-random-read,azure-random-read}.
Note that original documents may be stored in a single blob (e.g., delimited by line breaks) or in different blobs.
In each posting, \system records (blob name, offset, length) as part of a document identifier.

After constructing necessary structures, \system Searcher (described later shortly) can reside anywhere with an access to the cloud storage to serve search queries. Because of its configurable memory usage, it can be sufficiently portable to deploy on IoT and mobile devices, or it can inflate to support a humongous corpus on powerful machines. Perhaps one of interesting settings is the serverless deployment serving query requests, for example, function-as-a-service (FaaS) e.g. Google Cloud Functions \cite{gcp-function}, AWS Lambda \cite{aws-lambda}, and Microsoft Azure Functions \cite{azure-function}. Because of the minimal initialization, the deployment manager can quickly scale up or down based on the current demand across different corpuses.

\begin{table}[t]
\caption{Component-wise correspondence between Apache Lucene~\cite{lucene} and \system (this work). Skip list and postings list are sub-components of Lucene index. Likewise, MHT and superpost are sub-components of IoU Sketch.}
    \label{tab:term-mapping}
    
    \centering
    \small
    \begin{tabular}{ll}
        \toprule
         Apache Lucene & \system (Ours)  \\
        \midrule
        Lucene index (inverted index) & IoU Sketch \\
        $\cdot$ Skip Lists (term index) & $\cdot$ Multilayer Hash Table (MHT) \\
        $\cdot$ Postings List & $\cdot$ Super Postings List (superpost) \\
        \bottomrule
    \end{tabular}
    \vspace{-4mm}
\end{table}

\subsection{Statistical Inverted Index}
\label{sec:design:indexes}

This section describes systems aspect of our statistical (inverted) index.
We will mention core concepts introduced in \cref{sec:motivation:approximate:index},
with a focus on deployment and management perspectives.

Our statistical index is called \emph{IoU Sketch}; we named it IoU or \emph{intersection of unions} because our index construction involves unioning (or merging) multiple postings lists into one, and our search involves intersecting multiple postings lists.
IoU Sketch consists of two sub-components: a multi-layer hash table (MHT) and super postings lists (superposts). 
MHT is the multi-layer map we described in \cref{sec:motivation:approximate:index}.
Superposts are merged postings lists (also described in \cref{sec:motivation:approximate:index}).
MHT is downloaded and kept in memory when a certain corpus is searched for the first time.
Superposts are always kept on cloud storage.
To clarify their roles in existing contexts,
\cref{tab:term-mapping} makes one-to-one comparisons between Apache Lucene and ours.
Of course, their internal operations are different.





\subsection{Components and Workflow}
\label{sec:components}

\system consists of two components: Builder and Searcher. 
Builder creates and persists an index; Searchers uses the index for querying.

\paragraph{Builder Workflow} 
\system Builder (or simply, Builder) is a component that creates IoU Sketchs and persists them on cloud storage.
Builder creates a single IoU sketch per corpus.
Each IoU Sketch consists of superposts and MHT, both of which Builder creates and persists.
\cref{fig:design:airphant} lays out the steps to build IoU Sketch from profiling to optimization.


Builder's index creation starts when the user passes a corpus and configurations.
Builder uses a corpus-document parser to unwrap a blob into documents and generate postings that refer to their documents' byte ranges, allowing direct retrieval afterwards.
Builder then uses a document-word parser to extract words. The user can select both a corpus-document and a document-word parsers for each corpus. By default, a single blob may contain multiple documents; however, developers can override this with custom parsers.

These parsed documents are then profiled to collect statistics necessary for index building. 
\system Builder makes a single pass over all documents during profiling. The collected statistics include the total numbers of documents and words, document lengths, and document frequencies (i.e., the number of documents that contains a specific word). 
The statistics are used for IoU's structural optimization (\cref{section:iou-sketch,sec:airphant:stop-words}).

Based on the profile, Builder optimizes its IoU Sketch structure using the algorithm described in \cref{section:iou-sketch}.
The accuracy and memory requirements may also be specified.
Alternative to auto-tuning, users can also manually select the IoU Sketch structure, skipping both profiling and optimization steps.

\input{figures/design_airphant}

Builder first creates superposts.
\system Builder generates superposts from all documents. 
That is, a superpost is constructed for each bin, and a collection of all superposts are persisted.
Recall that each superpost is a list of merged postings, where
each posting
comprises (a blob name, an offset, and a length), which are used to locate actual documents.
The collection of superposts are concatenated into a single blob using a compaction encoding (\cref{section:doclist_compaction}).
This makes data management easier.


Next, Builder creates a MHT.
MHT stores the pointers to those superposts, which are compacted in a single blob.
To locate each superpoint, MHT's pointer has a blob name, byte offset, and byte size.
In addition, Builder stores seeds of hash functions (used for creating super posts) and other metadata in the same file.
This file is persisted 
as another blob.


\paragraph{Configuring Builder}

The user can configure \system in different ways. 
A storage driver specifies how to read a corpus.
Parsers specifies how to separate documents in a corpus and how to extract keywords from a document. 
The accuracy of IoU sketch can be set in terms of average number of irrelevant documents (i.e, false positive rate). 
The memory limit on MHT can also be placed.

\paragraph{Searcher} 

\system Searcher (or simply, Searcher) is a light-weight component that retrieves the documents containing search keywords.
Searcher relies on IoU Sketches.
Searcher performs two types of operations: initialization and querying.
Initialization happens only once per corpus.
Querying happens for each query.
The same diagram (\cref{fig:design:airphant}) illustrates the two procedures to fulfill queries: initialization and querying.

When a user opens a corpus, Searcher initializes itself from cloud-stored index structure; it retrieves
hash seeds and postings list pointers (for MHT), both of whose memory footprint is predictable and controllable 
via IoU Sketch configuration. It then reconstructs hash functions, and hence, MHT.

When a query arrives, Searcher hashes each word in the query to collect a set of pointers to postings list.
Using these pointers, Searcher concurrently fetches the corresponding postings lists from cloud storage and 
computes the intersection of postings lists (i.e., the final postings list). 
This is the only batch of concurrent requests for lookup.
Then, the documents identified by this final postings list is retrieved from cloud storage.
Finally, Searcher filters out irrelevant documents after fetching the documents. 
This filtering process is much fast compared to document-fetching.
The user can also fetch only top $K$ relevant documents, which is useful for lowing latencies (\cref{sec:airphant:top-k-query}).

\section{Statistical Inverted Index}
\label{sec:airphant}
\label{sec:indexing}


In this section, we detail our statistical inverted index (IoU Sketch).
First, we define IoU Sketch (\cref{section:iou-sketch}). 
Then, we detail corpus profiling (\cref{sec:airphant:profiling}) and superposts encoding (\cref{section:doclist_compaction}). 
Next, we introduce two  techniques to accelerate \system query: top-$K$ queries (\cref{sec:airphant:top-k-query}) 
and special handling of extremely common words (\cref{sec:airphant:stop-words}). 

\subsection{IoU Sketch} 
\label{section:iou-sketch}

IoU Sketch is the core index data structure that maps a keyword to a postings list. We first explain its data structure and interface as well as raw parameters. Second, given a corpus, we analyze the expected accuracy based on specific settings of raw parameters. Because one of raw parameters is not intuitive from user's perspective, IoU Sketch offers an alternative configuration based on memory constraint and desired accuracy. Using earlier analysis, it then automatically tunes its raw parameters to meet those constraints and optimize for search performance. Lastly, we further examine tightness of IoU Sketch's accuracy.

\paragraph{Data Structure} Internally, IoU Sketch is a $L$-layer hash table with $L$ different hash functions. Each bin in a table contains a superposts that is aggregated via insertion operations. Under IoU Sketch's hash functions, any given word is mapped to one bin per layer ($L$ bins in total) where the postings list of the word is guaranteed to be a subset of the superpost in each associated bin. Such guarantee eliminates false negative but entails false positive. IoU Sketch supports two main functionalities.
\begin{enumerate}
    \item \texttt{insert(word, document list)}: For each layer, it hashes the word to find its bin and update bin's superpost to its union with word's postings list.
    \item \texttt{query(word)}: It retrieves superposts from all layers; then, outputs the intersection of all superposts.
\end{enumerate}

\noindent
For illustration, \cref{fig:iou_example} shows IoU Sketch hash tables after inserting four words \texttt{w1}, \texttt{w2}, \texttt{w3}, and \texttt{w4} with different postings lists. Specifically, word \texttt{w2} is mapped to \texttt{(layer1, bin2)}, \texttt{(layer2, bin2)}, and \texttt{(layer3, bin1)} using hash functions from the three layers. It shares the same bin as \texttt{w3} in the first layer, \texttt{w4} in the second layer, and both \texttt{w1} and \texttt{w3} in the third layer. Each bin then stores the aggregated superpost of the words; for example, the superpost of \texttt{(layer1, bin2)} is the union of postings lists of words \texttt{w2} and \texttt{w3}: $\{ \texttt{d2}, \texttt{d3} \} \cup \{ \texttt{d2}, \texttt{d3}, \texttt{d4} \} = \{ \texttt{d2}, \texttt{d3}, \texttt{d4} \}$. Querying the word \texttt{w2} therefore results in a postings list $\{ \texttt{d2}, \texttt{d3}, \texttt{d4} \} \cap \{ \texttt{d2}, \texttt{d3}, \texttt{d4}, \texttt{d5} \} \cap \{ \texttt{d1}, \texttt{d2}, \texttt{d3}, \texttt{d4} \} = \{ \texttt{d2}, \texttt{d3}, \texttt{d4} \}$ which contains a false positive postings, \texttt{d4}. On the other hand, querying the word \texttt{w1} fortunately produces in the exact postings list $\{ \texttt{d1} \}$ despite the word sharing bins in second and third layers with other three words.

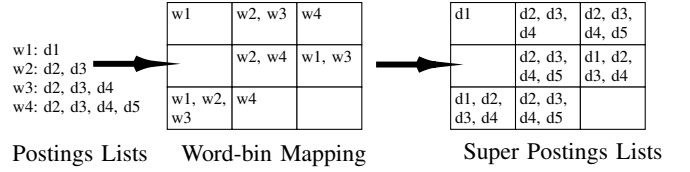
\begin{figure}[t]
\centering

\begin{subfigure}{\columnwidth}

\tikzset{every picture/.style={line width=0.40pt}} 

\begin{tikzpicture}[x=0.40pt,y=0.40pt,yscale=-1,xscale=1]

\draw  [fill={rgb, 255:red, 255; green, 255; blue, 255 }  ,fill opacity=1 ] (200,72) -- (261.33,72) -- (261.33,112) -- (200,112) -- cycle ;
\draw  [fill={rgb, 255:red, 255; green, 255; blue, 255 }  ,fill opacity=1 ] (200,112) -- (261.33,112) -- (261.33,152) -- (200,152) -- cycle ;
\draw  [fill={rgb, 255:red, 255; green, 255; blue, 255 }  ,fill opacity=1 ] (200,152) -- (261.33,152) -- (261.33,192) -- (200,192) -- cycle ;
\draw  [fill={rgb, 255:red, 255; green, 255; blue, 255 }  ,fill opacity=1 ] (261.33,72) -- (322.67,72) -- (322.67,112) -- (261.33,112) -- cycle ;
\draw  [fill={rgb, 255:red, 255; green, 255; blue, 255 }  ,fill opacity=1 ] (261.33,112) -- (322.67,112) -- (322.67,152) -- (261.33,152) -- cycle ;
\draw  [fill={rgb, 255:red, 255; green, 255; blue, 255 }  ,fill opacity=1 ] (261.33,152) -- (322.67,152) -- (322.67,192) -- (261.33,192) -- cycle ;
\draw  [fill={rgb, 255:red, 255; green, 255; blue, 255 }  ,fill opacity=1 ] (322.67,72) -- (384,72) -- (384,112) -- (322.67,112) -- cycle ;
\draw  [fill={rgb, 255:red, 255; green, 255; blue, 255 }  ,fill opacity=1 ] (322.67,112) -- (384,112) -- (384,152) -- (322.67,152) -- cycle ;
\draw  [fill={rgb, 255:red, 255; green, 255; blue, 255 }  ,fill opacity=1 ] (322.67,152) -- (384,152) -- (384,192) -- (322.67,192) -- cycle ;
\draw  [fill={rgb, 255:red, 255; green, 255; blue, 255 }  ,fill opacity=1 ] (468,72) -- (529.33,72) -- (529.33,112) -- (468,112) -- cycle ;
\draw  [fill={rgb, 255:red, 255; green, 255; blue, 255 }  ,fill opacity=1 ] (468,112) -- (529.33,112) -- (529.33,152) -- (468,152) -- cycle ;
\draw  [fill={rgb, 255:red, 255; green, 255; blue, 255 }  ,fill opacity=1 ] (468,152) -- (529.33,152) -- (529.33,192) -- (468,192) -- cycle ;
\draw  [fill={rgb, 255:red, 255; green, 255; blue, 255 }  ,fill opacity=1 ] (529.33,72) -- (590.67,72) -- (590.67,112) -- (529.33,112) -- cycle ;
\draw  [fill={rgb, 255:red, 255; green, 255; blue, 255 }  ,fill opacity=1 ] (529.33,112) -- (590.67,112) -- (590.67,152) -- (529.33,152) -- cycle ;
\draw  [fill={rgb, 255:red, 255; green, 255; blue, 255 }  ,fill opacity=1 ] (529.33,152) -- (590.67,152) -- (590.67,192) -- (529.33,192) -- cycle ;
\draw  [fill={rgb, 255:red, 255; green, 255; blue, 255 }  ,fill opacity=1 ] (590.67,72) -- (652,72) -- (652,112) -- (590.67,112) -- cycle ;
\draw  [fill={rgb, 255:red, 255; green, 255; blue, 255 }  ,fill opacity=1 ] (590.67,112) -- (652,112) -- (652,152) -- (590.67,152) -- cycle ;
\draw  [fill={rgb, 255:red, 255; green, 255; blue, 255 }  ,fill opacity=1 ] (590.67,152) -- (652,152) -- (652,192) -- (590.67,192) -- cycle ;
\draw [line width=2.25]    (130,130) -- (183,130) ;
\draw [shift={(187,130)}, rotate = 180] [color={rgb, 255:red, 0; green, 0; blue, 0 }  ][line width=2.25]    (17.49,-5.26) .. controls (11.12,-2.23) and (5.29,-0.48) .. (0,0) .. controls (5.29,0.48) and (11.12,2.23) .. (17.49,5.26)   ;
\draw [line width=2.25]    (397.5,130.67) -- (450.5,130.67) ;
\draw [shift={(454.5,130.67)}, rotate = 180] [color={rgb, 255:red, 0; green, 0; blue, 0 }  ][line width=2.25]    (17.49,-5.26) .. controls (11.12,-2.23) and (5.29,-0.48) .. (0,0) .. controls (5.29,0.48) and (11.12,2.23) .. (17.49,5.26)   ;

\draw (52,109) node [anchor=north west][inner sep=0.75pt]  [font=\scriptsize,xscale=0.9,yscale=0.9] [align=left] {w1: d1\\w2: d2, d3\\w3: d2, d3, d4\\w4: d2, d3, d4, d5};
\draw (202,75) node [anchor=north west][inner sep=0.75pt]  [font=\scriptsize,xscale=0.9,yscale=0.9] [align=left] {w1};
\draw (263.33,75) node [anchor=north west][inner sep=0.75pt]  [font=\scriptsize,xscale=0.9,yscale=0.9] [align=left] {w2, w3};
\draw (324.67,75) node [anchor=north west][inner sep=0.75pt]  [font=\scriptsize,xscale=0.9,yscale=0.9] [align=left] {w4};
\draw (263.33,115) node [anchor=north west][inner sep=0.75pt]  [font=\scriptsize,xscale=0.9,yscale=0.9] [align=left] {w2, w4};
\draw (324.67,115) node [anchor=north west][inner sep=0.75pt]  [font=\scriptsize,xscale=0.9,yscale=0.9] [align=left] {w1, w3};
\draw (202,155) node [anchor=north west][inner sep=0.75pt]  [font=\scriptsize,xscale=0.9,yscale=0.9] [align=left] {w1, w2,\\w3};
\draw (263.33,155) node [anchor=north west][inner sep=0.75pt]  [font=\scriptsize,xscale=0.9,yscale=0.9] [align=left] {w4};
\draw (470,75) node [anchor=north west][inner sep=0.75pt]  [font=\scriptsize,xscale=0.9,yscale=0.9] [align=left] {d1};
\draw (531.33,75) node [anchor=north west][inner sep=0.75pt]  [font=\scriptsize,xscale=0.9,yscale=0.9] [align=left] {d2, d3,\\d4};
\draw (592.67,75) node [anchor=north west][inner sep=0.75pt]  [font=\scriptsize,xscale=0.9,yscale=0.9] [align=left] {d2, d3,\\d4, d5};
\draw (531.33,115) node [anchor=north west][inner sep=0.75pt]  [font=\scriptsize,xscale=0.9,yscale=0.9] [align=left] {d2, d3,\\d4, d5};
\draw (592.67,115) node [anchor=north west][inner sep=0.75pt]  [font=\scriptsize,xscale=0.9,yscale=0.9] [align=left] {d1, d2,\\d3, d4};
\draw (470,155) node [anchor=north west][inner sep=0.75pt]  [font=\scriptsize,xscale=0.9,yscale=0.9] [align=left] {d1, d2,\\d3, d4};
\draw (531.33,155) node [anchor=north west][inner sep=0.75pt]  [font=\scriptsize,xscale=0.9,yscale=0.9] [align=left] {d2, d3,\\d4, d5};
\draw (51.67,204.67) node [anchor=north west][inner sep=0.75pt]  [xscale=0.9,yscale=0.9] [align=left] {Postings Lists};
\draw (211.33,204.33) node [anchor=north west][inner sep=0.75pt]  [xscale=0.9,yscale=0.9] [align=left] {Word-bin Mapping};
\draw (478,203) node [anchor=north west][inner sep=0.75pt]  [xscale=0.9,yscale=0.9] [align=left] {Super Postings Lists};

\end{tikzpicture}

\end{subfigure}

\caption{Example of IoU Sketch for a corpus with 5 documents and 4 distinct words. Word-bin mapping shows sets of words in corresponding IoU Sketch's bins. 
Super postings lists are the merged postings lists for those sets of words.}
\label{fig:iou_example}

\vspace{-3mm}
\end{figure}

Recall there are $n$ documents and let $B$ be the total number of bins across all layers. To persist, IoU Sketch only requires superposts and hash seeds to reconstruct itself, making its worst case storage complexity $O(\min\{ nmL, Bn+L \})$ if documents have $m$ words on average. To achieve single-cycle retrievals, \system Builder split IoU Sketch into superposts and MHT. It then stores superposts on cloud storage, so that MHT is $L$ layers of hash tables where each entry has a pointer to the corresponding superpost. Therefore, \system Searcher needs to hold $O(B)$ of memory for MHT ($O(L)$ hash seeds and $O(B)$ bin pointers where $L \ll B$).


\cref{fig:false-positives} empirically demonstrates that IoU Sketch achieves significantly fewer false positives compared to a hash table ($L = 1$) when $B$ is fixed. As $L$ increases from $1$, the number of false positives decreases rapidly. After a certain value, the allocated bins are divided into too many layers, resulting in a smaller number of bins per layer, and consequently, a higher selectivity and more false positives. Not only do these observations justify the use of multiple hash tables over a single one, but it also suggests that there is an optimal choice of $L$ depending on both the corpus and the choice of $B$.

\paragraph{Expected False Positives}
For a family of pairwise independent hash functions, IoU Sketch selects $(h_l)_{l=1}^L$ where each $h_l$ is the hash function for $l$-th layer. Let $W_i$ be the set of distinct words in $i$-th document with its size being $|W_i|$, and $W$ be the set of all words. Suppose $B$ is given and assume $B$ is divisible by $L$ for simplicity, the probability that $i$-th document is a false positive in a query of any irrelevant word $w$ is $q_i(L) = q_i(L; B, \{W_i\}_{i=1}^n)$. Independence between $w$ and $q_i$ is a result of pairwise independent hash family. \cref{eq:ind-fp} also shows the approximation $\qh_i(L)$ whose properties leads to an efficient optimization.
\begin{equation} \label{eq:ind-fp}
    q_i(L) = \left[ 1 - \left( 1 - \frac{1}{B/L} \right)^{|W_i|} \right]^L \approx \left[ 1 - e^{\frac{-|W_i| L}{B}} \right]^L = \qh_i(L)
\end{equation}

Assume a query word distribution $w \sim \text{Cat}(W, \vec{p})$ where $p_w$ is the prior probability of the word $w$ in a query, \cref{eq:fploss} describes the expected number of false positives over all words whose unit is count per query. This is the primary objective function to tune IoU Sketch. For brevity, we write $F(L) = F(L; B, \{W_i\}_{i=1}^n)$ and define its approximation $\Fh(L)$ similarly using $\qh_i(L)$ in place of $q_i(L)$. Here $c_i = \sum_{w \in W \setminus W_i} p_w$ is the probability of words not contained in $i$-th document and acts as a linear combination coefficient in $F$. \cref{fig:false-positives} asserts the resemblance between this formula and empirical observations. In fact, later, we confirms that the observed number of false positives is highly concentrated around this expectation.
\begin{equation} \label{eq:fploss}
    \begin{split}
        F(L) = 
        \sum_{i=1}^n \sum_{w \in W \setminus W_i} p_w q_i(L) = 
        \sum_{i=1}^n c_i q_i(L)
    \end{split}
\end{equation}

\input{figures/method_false_positives}

Even though $L$ is a discrete variable, we extend its domain to a continuous one $L \in \mathbb{R}, 1 \leq L \leq B$ to study richer characteristics of $F(L)$. In particular, the extension endows us with the derivative $\fh(L) = \frac{d}{dL} \Fh(L) = \sum_{i=1}^n c_i \qh_i(L)$. We condense the formula by substituting in $z_i(L) = 1 - \exp\left( -|W_i| L / B \right)$. The shift of focus towards the approximation eases the analysis and leads to an efficient algorithm to optimize IoU Sketch structure.
\begin{equation} \label{eq:ind-fp-grad}
    \qh_i'(L) = z_i(L)^{L-1} \left[ z_i(L) \ln z_i(L) - (1 - z_i(L)) \ln (1 - z_i(L)) \right]
\end{equation}

\paragraph{Optimization on Number of Layers}
The overall performance improves when the number of layers is smaller since a query would fetch and intersect a fewer superposts. In addition, a larger number of bins would replicate more postings across layers, further increasing the index storage size. Consequently, we are interested in minimizing the number of layers given constraints on the number of bins $B$ and false positives $F_0$. In other words, the optimization problem (\cref{eq:opt-problem}) finds the smallest number of layers that the $(B, L)$ IoU Sketch has a fewer expected number of false positives than $F_0$.
\begin{equation} \label{eq:opt-problem}
    L^{*} = \min L \quad \text{s.t. } 1 \leq L \leq B, F(L; B, \{W_i\}_{i=1}^n)) \leq F_0
\end{equation}

Unfortunately, $F(L)$ is non-convex and can contain multiple minimizers. Nevertheless, an analysis on its approximation $\Fh(L)$ reveals three important characteristics as the building blocks for Algorithm~\ref{algo:layer-opt}. First, there is a region of fast optimization covering practical $F_0$ values (\cref{lemma:fploss-left-fast}). Secondly, although $L$ can be as large as $B$, we only need to search in a much smaller interval (\cref{lemma:fploss-right-done}). Finally, there is a relatively cheap lower bound which allows us to quickly check the feasibility (\cref{lemma:fploss-lowerbound}).

With these lemmas in mind, Algorithm~\ref{algo:layer-opt} first validates the proposed constraints $B$ and $F_0$ with the lower bound. It then determines whether $L$ falls in the fast or slow region and picks the optimization routine accordingly. For the fast region where $\Fh(L)$ is decreasing, it performs a binary search to find the smallest $L$ in the range $[1, L_{\min}]$. On the other hand, the slow region in the range $[L_{\min}, L_{\max}]$ does not guarantee such monotonicity, so the algorithm iteratively attempts increasing values of $L$ until the constraint is met. If either the lower bound checking or iterative search fails, it is impossible to find a $L$ that satisfies the constrains and so the algorithm rejects.

\begin{lemma} \label{lemma:fploss-lowerbound}
    $L_i^{*} = \argmin_L \qh_i(L) = \frac{B}{|W_i|} \ln 2$. It immediately follows that $\qh_i(L_i^{*}) = 2^{-L_i^{*}}$ and so $\Fh(L) \geq \sum_{i=1}^n c_i 2^{-L_i^{*}}$.
\end{lemma}
\begin{proof}
    From \cref{eq:ind-fp-grad}, the minimizer $L_i^{*}$ satisfies $\qh_i'(L_i^{*}) = 0$. Note that the left factor is always positive, so it must be the case that the right factor is zero or equivalently $z_i(L_i^{*}) = 1/2$. Therefore, $\exp\left( -|W_i| L_i^{*} / B \right) = L_i^{*}$; in other words, $L_i^{*} = \frac{B}{|W_i|} \ln 2$. Substitute this minimizer to \cref{eq:ind-fp} and \cref{eq:fploss} to produce the two results later.
\end{proof}

\begin{remark}
Because $F(L) > \Fh(L)$ for $1 \leq L \leq B$, we have also derived a lower bound $F(L) > \sum_{i=1}^n c_i 2^{-L_i^{*}}$, validating the feasibility check in \cref{algo:layer-opt}.
\end{remark} 

\begin{lemma} \label{lemma:fploss-left-fast}
    For $L < \min_i L_i^{*} = L_{\min}$, expected false positive is strictly and exponentially decreasing $\fh(L) < 0$ and $\Fh(L) = O\left( \frac{n}{2^L} \right)$.
\end{lemma}
\begin{proof}    
    Notice that if $L < L_i^{*}$, $z_i(L) < 1/2$. Using \cref{eq:ind-fp-grad}, we prove the strictly decreasing property: $\qh_i'(L) < 2^{-L-1} \ln \frac{z_i(L)}{1 - z_i(L)} < 0$. Similarly from $z_i(L) < 1/2$, we also have that $\qh_i(L) < 2^{-L}$ by \cref{eq:ind-fp}; therefore, if $L < \min_i L_i^{*}$, $\qh_i(L) < 2^{-L}$ for all $i \in [n]$. Subsequently, the expected false positive is also exponentially decreasing $\Fh(L) = \sum_{i=1}^n c_i \qh_i(L) < \sum_{i=1}^n c_i 2^{-L} \leq n 2^{-L}$.
\end{proof}

\begin{remark}
The region $[1, L_{\min}]$ covers a wide range of $F_0$. Even in the worst case where $c_i = 1$, the region covers the expected number of false positives down as low as $n 2^{-L_{\min}} = n 2^{-\frac{B}{\max_i |W_i|} \ln 2}$, that is, $B \geq \frac{1}{\ln 2} \times \max_i |W_i| \times \log_2 \frac{n}{F_0}$ surely enables fast optimization via binary search on the strictly decreasing function. Nonetheless, \cref{algo:layer-opt} measures a tighter expected false positive lower bound of the region $F(L_{\min})$ to decide whether to use fast optimization.
\end{remark} 

\begin{lemma} \label{lemma:fploss-right-done}
    For $L > \max_i L_i^{*} = L_{\max}$, expected false positive is strictly increasing $\fh(L) > 0$.
\end{lemma}
\begin{proof}
    If $L > L_i^{*}$, $z_i(L) > 1/2$. Together with \cref{eq:ind-fp-grad}, it implies that $\qh_i'(L) > 2^{-L-1} \ln \frac{z_i(L)}{1 - z_i(L)} > 0$.
\end{proof}

\begin{algorithm}[t]
  \caption{\small Number of Layers Minimization} 
  \label{algo:layer-opt}
  \small
  
  \DontPrintSemicolon

  \KwInput{Number of bins $B$, expected false positive $F_0$, sets of all distinct words $\{W_i\}_{i=1}^n$, query word distribution $\text{Cat}(W, \vec{p})$}
  \KwOutput{Minimum number of layer $L^{*}$ or rejection}

  \If{$\sum_{i=1}^n c_i 2^{-L_i^{*}} \leq F_0$} {
    \If{$F(L_{\min}) \leq F_0$} {
        $L^{*} \leftarrow$ binary search $L \in [1, L_{\min}]$
    } \ElseIf{iterative search $L \in [L_{\min}, L_{\max}]$ succeeds} {
        $L^{*} \leftarrow$ the result from iterative search
    }
  }
  return $L^{*}$ if assigned, otherwise reject
\end{algorithm}

\paragraph{False Positive Guarantee}
For fixed numbers of bins $B$ and layers $L$, each false positive from $i$-th document on irrelevant query word $w$ is a multiple of Bernoulli random variable, i.e. $x_{i, w} = p_w b_i$ where $b_i \sim \text{Bern}(q_i(L))$. Since $x_{i, w} \in [0, p_w]$ and $E[x_{i, w}] = p_w q_i(L)$, Hoeffding's inequality guarantees that the observed number of false positives $X = \sum_{i=1}^n \sum_{w \in W \setminus W_i} x_{i, w}$ does not deviate more than the expectation $E[X] = F(L)$ by $\varepsilon$ with probability at least $1 - \delta$.
\begin{equation} \label{eq:fploss-highp}
    \begin{split}
        \delta = \Pr[X \geq F(L) + \varepsilon]
            &\leq \exp \left( - 2 \varepsilon^2 / \sigma_X^2 \right)
    \end{split}
\end{equation}
where $\sigma_X^2 = \sum_{i=1}^n \sum_{w \in W \setminus W_i} p_w^2$. Thus, the deviation is bounded with $\varepsilon \leq \sqrt{ \frac{1}{2} \sigma_X^2 \ln \frac{1}{\delta} }$. In the worst case where very few query words are irrelevant to all documents and dominate the distribution ($p_w \rightarrow 1$) making $\sigma_X^2 \rightarrow n$, the deviation can possibly be as large as $\varepsilon \leq \sqrt{ \frac{n}{2} \ln \frac{1}{\delta} } = O(\sqrt{n})$; however, a memoization technique such as query caching suffices to solve this case. It is worth noting that, in a typical case when there are many irrelevant query words with similar probability, the deviation would instead shrink as the number of words increases: $\varepsilon \leq \sqrt{ \frac{n}{2 |W|} \ln \frac{1}{\delta} } = O\left(\sqrt{\frac{n}{|W|}}\right)$.

As a point of reference, \cref{tab:corpus-stats} summarizes corpus-dependent coefficients $\sigma_X$ for each corpus in the experiment assuming query words distribute uniformly among relevant words found in corpus. To reiterate, this analysis only ensures the concentration of the observed number of false positives around $F(L)$. The number of false positives itself can be reduced by adjusting IoU Sketch's structure $(B, L)$.

\subsection{Corpus Profiling}
\label{sec:airphant:profiling}

Although the optimization formulation factors in any categorical query word distribution $w$, \system assumes a uniform distribution by default; in other words, a query equally likely contains words in the corpus or $p_w = 1 / |W|$. While no further evidence support nor deny such choice, it is potentially simplistic. Other possible apparent choices with their profiling are: (a) $p_w = \text{occurrences}(w)$ by profiling word occurrences and total number of words, and (b) user-provided or statistical prior $p_w$ with no profiling. One might also consider assigning non-zero $p_{w'}$ where $w' \notin W$. We defer seeking more suitable distributions to future studies and let this implementation be a case study. Subsequently, \system Builder's profiling then counts the number of distinct words in the corpus as well as the numbers of distinct words within each document.

\subsection{Superpost Compaction} 
\label{section:doclist_compaction}
\label{sec:airphant:compaction}

\system Builder implements a simple superpost compaction to avoid creating too many tiny or a few huge files, and to allow single-cycle retrievals in \system Searcher. Previously hinted in earlier sections, the compaction comprises of two components: a header block and superpost blocks.

Each of the superpost blocks stores serialized multiple superposts consecutively. \system internally uses Protocol Buffers to serialize superposts to byte arrays. While indexing, \system Builder keeps track of each superpost location and builds the dictionary of bin pointers. Given the superpost block structure, each bin pointer need to represent block ID, offset, and byte length to retrieve the superpost's bytes in a single round-trip. In addition, \system compresses repeated strings within postings into integer keys. The compression reduces the number of bytes per superpost to be downloaded which speeds up query overall.

\system Builder persists these bin pointers and string compression table along side with hash seeds and other metadata in the header block. The hash seeds are collected from the hash function in IoU Sketch; in other words, they concisely represents IoU Sketch mapping. It is this header block that is loaded on \system Searcher initialization.

\subsection{Top-K Query}
\label{sec:airphant:top-k-query}

Instead of retrieving all relevant documents in a query, \system Searcher supports fetching at least $K$ relevant documents. Top-$K$ query enables pagination for providing a quick view or batch processing. Thanks to IoU Sketch's false positive guarantee, that is, the approximated superpost contains $F_0$ irrelevant documents on average, \system Searcher can sample a subset of the superpost to fetch from. Suppose the superpost contains $R$ postings, if $K \geq R - F_0$, then \system Searcher fetches all $R$ documents. Otherwise, each posting corresponds to a relevant document with Bernoulli distribution $Bern(p = 1 - F_0 / R)$. With probability at least $1 - \delta$, solving a quadratic inequality after applying Hoeffding's inequality guarantees that sampled postings of size $R_K$ (\cref{eq:topk_sample_size}) comprise of at least $K$ relevant documents.
\begin{equation} \label{eq:topk_sample_size}
    R_K = \left\lceil \frac{2 p K + \frac{1}{2} \ln \frac{1}{\delta} + \sqrt{(2 p K + \frac{1}{2} \ln \frac{1}{\delta})^2 - 4 p^2 K^2}}{2 p^2}  \right\rceil
\end{equation}

\subsection{Common Words}
\label{sec:airphant:stop-words}

Common words are keywords that are contained in many documents in the corpus. Some information retrieval systems assign them as stop words and filter them out during all retrieval steps. In contrary, \system supports searching for common words. The challenge is that, merging their large postings lists into IoU Sketch's bins would deteriorate performance on other keyword query as well. As a workaround, \system sets aside $1\%$ of the bins to store the exact postings lists of most common words. For example, if $B = 10^5$, \system would use $99,000$ bins for IoU Sketch and $1,000$ bins to carry $1,000$ most common words' postings lists. We use the same superpost compaction for these postings lists. 

\subsection{Applicable Queries}

Although IoU Sketch only natively supports queries on a single term, we can adapt it to accelerate other classes of queries. For one, like an inverted index, IoU Sketch naturally generalizes to Boolean queries \cite{Manning2008IntroIR}. Let $Q(w)$ be the superpost from querying IoU Sketch with word $w$. IoU Sketch executes any Boolean query by distributing its query function to each term predicate $Q(\bigvee_i \bigwedge_j w_{ij}) = \bigcup_i \bigcap_j Q(w_{ij})$. In particular, intersection and union operators apply to superposts where the former reduces false positives and the latter adds. Furthermore, regular expression (RegEx) can benefit from IoU Sketch as inverted index by considering indexing $N$-grams as shown in RegEx engines 
\cite{Qiu2020EfficienRegexPosInvInd} \cite{Cho2002RegexIndex}. These engines use an inverted index as a filter to avoid a full corpus scan, and later match the remaining documents with the RegEx to remove false positives. Hence, superpost's false positives do not affect the final correctness.

\subsection{Built-in Replication for Reliability}

During its query execution, IoU Sketch retrieves superposts from all $L$ layers simultaneously in parallel I/O requests. As a result, the slowest retrieval among $L$ requests defines the total query latency, exposing IoU Sketch to the Long Tail Problem if there is any extreme variation during the I/Os such as dormant storage or network congestion \cite{Gill2020TailsIT}. Nonetheless, IoU Sketch's multi-layer structure comes to the rescue as a built-in replication mechanism. The core idea is that IoU Sketch can simply discard the pending slow request and return the available-but-suboptimal superpost. The simplest mitigation is then to set a timeout before aborting the trailing request. Alternatively, IoU Sketch can enhance its reliability by overestimating the number layers, say $L_{+}$, adding magnitudes of accuracy as a consequence. The more layers IoU Sketch overestimates, the stronger its reliability against long-tail requests becomes. During a query, IoU Sketch would then submit $L_{+}$ I/O requests but only wait for any $L$ successful retrievals. We refer to \cite{Wang2019EfficientSR} for more sophisticated techniques and analysis frameworks.







\section{Experiments}
\label{sec:exp}

We have conducted empirical studies to evaluate \system's performance.
We observe that:

\begin{enumerate}
    \item \system outperforms Lucene, Elasticsearch, SQLite, and HashTable with upto $8.97\times$, $113.39\times$, $3.15\times$, and $378.59\times$ faster response respectively (\cref{sec:experiment:e2e}).
    
    \item \system competitively controls the effect of physical distance between compute and storage compared to baselines (\cref{sec:experiment:cross_region}).
    
    \item Our latency analysis shows that baselines are slower due to either time spent being blocked waiting for network response or downloading large amount of data. \system reduces both portions of time at the same time (\cref{sec:experiment:query_breakdown}).
    
    \item With the separation of compute and storage, \system is more cost-efficient than a Elasticsearch deployed with local persistence when the workload is skewed (high or rare peak throughput) and/or the size of indexed data is large (\cref{sec:experiment:cost_compare}).
    
    
    
    \item Our last experiments justify IoU Sketch analysis, showing that IoU Sketch provides a trade-off between memory usage and latency and that its structure optimization is crucial. (\cref{sec:experiment:trafeoff-iou})
\end{enumerate}

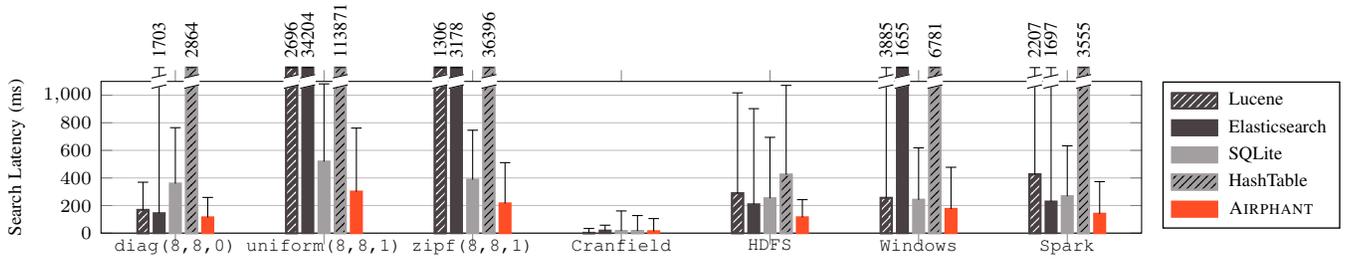
\begin{figure*}[t]
  \centering

\pgfplotsset{accfig/.style={
    width=0.85\linewidth,
    height=36mm,
    ybar,
    bar width=1.45mm,
    xmin=0.5,
    xmax=7.5,
    xtick={1, 2, ..., 7},
    xticklabels={\texttt{diag(8,8,0)}, \texttt{uniform(8,8,1)}, \texttt{zipf(8,8,1)}, \texttt{Cranfield}, \texttt{HDFS}, \texttt{Windows}, \texttt{Spark}},
    ylabel near ticks,
    ylabel style={align=center},
    ymin=0,
    ymax=1100,
    ytick={0, 200, ..., 1100},
    xticklabel style={yshift=2mm},
    xlabel style={yshift=0mm},
    legend style={
        at={(1.02,1.0)},anchor=north west,column sep=2pt,
        draw=black,fill=white,line width=.5pt,
        /tikz/every even column/.append style={column sep=5pt}
    },
    legend cell align={left},
    legend columns=1,
    area legend,
    clip=false,
    every axis/.append style={font=\scriptsize},
    minor grid style=lightgray,
    ymajorgrids,
}}

    \begin{tikzpicture}

    \def\yclip{1200}
    \def\ycliptop{1200}
    \def\ycliplowerlower{1050}
    \def\ycliplower{1075}
    \def\yclipupper{1125}
    \def\yclipupperupper{1150}
    \newcommand\xmid[2]{#1 + #2}
    \newcommand\xlower[3]{\xmid{#1}{#2} - #3}
    \newcommand\xupper[3]{\xmid{#1}{#2} + #3}
    
    \newcommand{\slabClipNumber}[4]{%
        \draw[fill=white,draw=white] (axis cs: \xlower{#1}{#2}{#3},\yclipupper) -- (axis cs: \xupper{#1}{#2}{#3},\yclipupperupper) -- (axis cs: \xupper{#1}{#2}{#3},\ycliplower) -- (axis cs: \xlower{#1}{#2}{#3},\ycliplowerlower) -- cycle;%
        \draw[draw=black] (axis cs: \xlower{#1}{#2}{#3},\yclipupper) -- (axis cs: \xupper{#1}{#2}{#3},\yclipupperupper);%
        \draw[draw=black] (axis cs: \xupper{#1}{#2}{#3},\ycliplower) -- (axis cs: \xlower{#1}{#2}{#3},\ycliplowerlower);%
        \node[anchor=west, rotate=90] at (axis cs: \xmid{#1}{#2},\ycliptop) {#4};%
    }

    \begin{axis}[accfig,
        ylabel=Search Latency (ms)]
    
    \addplot[thick,draw=vintageblack,fill=vintageblack,
    postaction={
        pattern=north east lines,
        pattern color=white,
    }]
    plot [error bars/.cd, y dir=plus,y explicit,error bar style={color=black}]
    table[x=x,y=y,y error expr=\thisrowno{3}-\thisrowno{2}] {
    x c y    y+
    1 0 167.8090757 369.065725
    2 0 {\yclip}  {\yclip}
    3 0 {\yclip}  {\yclip}
    4 0 1.681455812 33.711669
    5 0 288.4062741 1016.93374
    6 0 254.8345557 {\yclip}
    7 0 426.2369877 {\yclip}
    };
    \slabClipNumber{2}{-0.22}{0.05}{2696};
    \slabClipNumber{3}{-0.22}{0.05}{1306};
    \slabClipNumber{6}{-0.22}{0.05}{3885};
    \slabClipNumber{7}{-0.22}{0.05}{2207};
    
    \addplot[thick,draw=vintageblack,fill=vintageblack]
    plot [error bars/.cd, y dir=plus,y explicit,error bar style={color=black}]
    table[x=x,y=y,y error expr=\thisrowno{3}-\thisrowno{2}] {
    x c y    y+

    1 1 143.99644 {\yclip}
    2 1 {\yclip}  {\yclip}
    3 1 {\yclip}  {\yclip}
    4 1 18.14514992 57.571147
    5 1 208.6094791 902.470908
    6 1 {\yclip}  {\yclip}
    7 1 228.580414  {\yclip}
    };
    \slabClipNumber{1}{-0.11}{0.05}{1703};
    \slabClipNumber{2}{-0.11}{0.05}{34204};
    \slabClipNumber{3}{-0.11}{0.05}{3178};
    \slabClipNumber{6}{-0.11}{0.05}{1655};
    \slabClipNumber{7}{-0.11}{0.05}{1697};
    
    \addplot[thick,draw=vintageblack!50,fill=vintageblack!50,]
    plot [error bars/.cd, y dir=plus,y explicit,error bar style={color=black}]
    table[x=x,y=y,y error expr=\thisrowno{3}-\thisrowno{2}] {
    x c y    y+
    1 2 359.275471  764.487312
    2 2 518.7315046 1081.898583
    3 2 387.5628978 746.111403
    4 2 15.06970779 160.938858
    5 2 252.9398227 694.395288
    6 2 242.0121507 618.240363
    7 2 267.5795218 632.960519
    };
    
    \addplot[thick,draw=vintageblack!50,fill=vintageblack!50,
    postaction={
        pattern=north east lines,
        pattern color=black,
    }]
    plot [error bars/.cd, y dir=plus,y explicit,error bar style={color=black}]
    table[x=x,y=y,y error expr=\thisrowno{3}-\thisrowno{2}] {
    x c y    y+
    1 3 {\yclip}  {\yclip}
    2 3 {\yclip}  {\yclip}
    3 3 {\yclip}  {\yclip}
    4 3 15.47910109 127.370042
    5 3 423.2580287 1071.319562
    6 3 {\yclip}  {\yclip}
    7 3 {\yclip}  {\yclip}
    };
    \slabClipNumber{1}{0.11}{0.05}{2864};
    \slabClipNumber{2}{0.11}{0.05}{113871};
    \slabClipNumber{3}{0.11}{0.05}{36396};
    \slabClipNumber{6}{0.11}{0.05}{6781};
    \slabClipNumber{7}{0.11}{0.05}{3555};
    
    
    \addplot[thick,draw=vintageorange,fill=vintageorange,]
    plot [error bars/.cd, y dir=plus,y explicit,error bar style={color=black}]
    table[x=x,y=y,y error expr=\thisrowno{3}-\thisrowno{2}] {
    x c y    y+
    1 4 114.034835  258.821839
    2 4 300.7736378 761.661508
    3 4 215.1467519 511.057993
    4 4 13.4467168  105.530441
    5 4 114.439757  242.633805
    6 4 175.3187697 477.434309
    7 4 140.4061745 373.191297
    };

    \addlegendentry{Lucene}
    \addlegendentry{Elasticsearch}
    \addlegendentry{SQLite}
    \addlegendentry{HashTable}
    \addlegendentry{\system}
      
    \end{axis}
    \end{tikzpicture}
    

  \vspace{-2mm}
  \caption{End-to-end search latencies of indexes on different datasets. 
  Solid bars show average latencies; the upper error bars show 99th percentiles. Latencies over 1.1 seconds are truncated and labeled. 99 percentiles whose means are over the limit are omitted.}
  
  \label{fig:e2e_search}
  \vspace{-3mm}
\end{figure*}

\subsection{Environment Setup}
\label{sec:experiment:setup}

Our experiments used Google Cloud Platform. 
It uses GCP Cloud Storage~\cite{gcp-storage} as the cloud storage for index structure persistence. To provide file-system interface for all benchmarks, we connect all necessary storage buckets to a directory using Cloud Storage FUSE~\cite{gcp-fuse} (\texttt{gcsfuse}) adaptor with no limit on rate of operations. We allocate two VM instances \texttt{n2-highmem-32} (32 vCPUs, 256 GB memory, 1 TB SSD persistent disk) and \texttt{e2-small} (2 vCPUs, 2 GB memory, 10 GB default boot disk) for indexing and query benchmarking respectively. All experiments run from September to November 2021.

\begin{table}
\caption{Corpus Statistics. \#documents: number of documents. \#terms: number of distinct words. \#words: total number of words across all documents. $\sigma_X$: corpus-dependent coefficient discussed in \cref{section:iou-sketch}. }
    \label{tab:corpus-stats}
    
    \centering
    \small
    \begin{tabular}{lrrrr}
        \toprule
        Corpus & \#documents & \#terms & \#words & $\sigma_X$ \\
        \midrule
        \texttt{diag(8,8,0)} & $10^8$ & $10^8$ & $10^8$ & $1.00$ \\
        \texttt{unif(8,8,1)} & $10^8$ & $1.0 \times 10^8$ & $1.0 \times 10^{9}$ & $1.00$ \\
        \texttt{zipf(8,8,1)} & $10^8$ & $5.0 \times 10^7$ & $9.5 \times 10^{8}$ & $1.41$ \\
        \texttt{Cranfield} & $1.4 \times 10^3$ & $5.3 \times 10^3$ & $1.2 \times 10^{5}$ & $0.51$ \\
        \texttt{HDFS} & $1.1 \times 10^7$ & $3.6 \times 10^6$ & $1.4 \times 10^{8}$ & $1.77$ \\
        \texttt{Windows} & $1.1 \times 10^8$ & $8.3 \times 10^5$ & $1.7 \times 10^{9}$ & $11.73$ \\
        \texttt{Spark} & $3.3 \times 10^7$ & $5.2 \times 10^6$ & $3.5 \times 10^{8}$ & $2.53$ \\
        \bottomrule
    \end{tabular}

\vspace{-2mm}
\end{table}

\paragraph{Datasets} There are 4 corpuses to benchmark these index systems. \texttt{Cranfield} (Cranfield 1400)~\cite{Cleverdon1967} is a small corpus of 1398 documents, each contains abstracts from aerodynamics research papers. \texttt{HDFS}~\cite{Xu2009}, \texttt{Windows}, and \texttt{Spark} are system logs in their corresponding systems collected by Loghub~\cite{He2020}.

In addition, we use 3 types of synthetic datasets whose size is configurable. We denote the size of each synthetic dataset using a tuple $(\log_{10} n_d, \log_{10} n_w, \log_{10} n_l)$ for its numbers of documents $n_d$, words $n_w$, and words per documents $n_l$. \texttt{diag} is a dataset where each document $i$ contains only one word $w_{i}$. As a result, \texttt{diag} always has $n_l = 1$. \texttt{unif} is a dataset where each word in a document is uniformly sampled from the $n_w$-word dictionary. \texttt{zipf} is similar to \texttt{unif} but uses a Zipfian distribution with the exponent equal to $1.07$. In other word, each word in a document is equal to $w_j$ with probability proportional to $1 / j^{1.07}$. Note that \texttt{unif} and \texttt{zipf} can under-generate the actual set of distinct words form $n_w$ due to Coupon collector's problem~\cite{Flajolet1992}. \cref{tab:corpus-stats} summarizes corpuses' statistics that are used in \cref{sec:experiment:e2e}.

\paragraph{Baselines}

    We compare \system to Lucene 7.4.0~\cite{lucene}, Elasticsearch 7.15.1~\cite{elasticsearch}, SQLite 3.34.0~\cite{sqlite}, and na\"ive hash table.
    Lucene is an information retrieval library. 
    Elasticsearch is a search and analytics engine. 
    We benchmark their efficiency in matching exact keywords. 
    We first parse keywords out similarly to other baselines and feed parsed documents to a text field using Elasticsearch's \texttt{whitespace} analyzer and Lucene's \texttt{WhitespaceAnalyzer} along with the document's posting. To benchmark Elasticsearch, we mount a Searchable Snapshot~\cite{elasticsearchSearchableSnapshot} onto an Elasticsearch empty instance to speed up its initial latency. 

    SQLite is a light database we choose as a practical B-tree implementation. We first create a two-column table consisting of keyword column and postings column to mimic the inverted index dictionary. We then build SQLite's B-tree index~\cite{SQLit2021Btree} on the keyword column as its term index and store its database file on the cloud-mounted directory. In each query, after retrieving the postings, SQLite reuses the same document retrieval routine from \system.
    
    Lastly, HashTable refers to an inverted index that stores postings lists according to their corresponding terms' hashes. It is equivalent to IoU Sketch with the only exception that it has a single layer $L = 1$. Other relevant configurations such as the total number of bins and common word bins are identical to IoU Sketch.
    All postings inserted in all baselines are compressed in the same way as in \system (\cref{sec:airphant:compaction}).



\paragraph{Parameters} If otherwise specified, we set the number of bins $B = 10^5$, accuracy constraint $F_0 = 1$, and probability of top-$K$ query failure $\delta = 10^{-6}$ for $K = 10$. Remind that we allocate 1\% of the bins to store postings lists of most common words. We similarly set top-$K$ query to other baselines accordingly. The top-$K$ query failure never occurs during the experiment due to the conservative setting which selects about $23$ samples to answer top-$10$ query. The accuracy choice leads to optimal number of layers $L^{*}$ in the at most 3 layers depending on the corpus. 
The number of bins results in \system Search's runtime size about $2$ MB.
We use $32$ threads to concurrently read any data from the cloud storage.

\subsection{End-to-end Search Performance}

\paragraph{Within Region}
\label{sec:experiment:e2e}

    To demonstrate \system's performance, we measure end-to-end search latency for each query. \cref{fig:e2e_search} shows mean and 99th-percentile latencies.
    \system query executions are $1.45\times$ to $8.97\times$ faster than Lucene's on average (except \texttt{Cranfield} dataset where Lucene is $8.00\times$ faster), $1.09\times$ to $113.39\times$ faster than Elasticsearch's, $1.12\times$ to $3.15\times$ faster than SQLite's, and $1.15\times$ to $378.59\times$ faster than HashTable's. 
    HashTable is slow because it
    spends the majority of its latency to filter out false-positive documents, which highlights the strength of IoU Sketch's multi-layer structure. 

    From these benchmarks, we see that \system operates at less than 300 ms on average. According to~\cite{Brutlag2008UserPA, Arapakis2021ImpactOR}, web search users start to notice the latency when it is on the order of seconds, and so, they are tolerable with \system latency. To put its latency into a perspective, \system would only add a factor to the 250 ms ``speed-of-light web search latency''.

\paragraph{Cross Region}
\label{sec:experiment:cross_region}

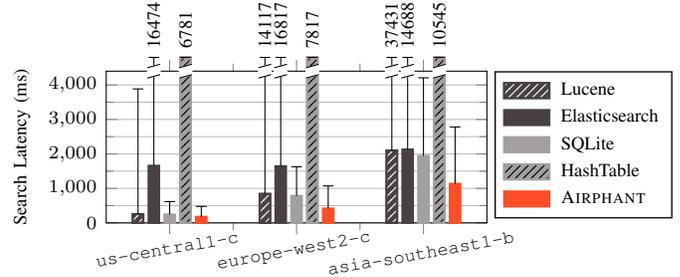
\begin{figure}[t]
  \centering

\pgfplotsset{
    accfig/.style={
        width=0.75\linewidth,
        height=36mm,
        ybar,
        bar width=1.4mm,
        xmin=0.5,
        xmax=3.5,
        xtick={1, 2, 3},
        xticklabels={
            \texttt{us-central1-c}, 
            \texttt{europe-west2-c}, 
            \texttt{asia-southeast1-b}, 
        },
        xticklabel style={rotate=10},
        ylabel near ticks,
        ylabel style={align=center},
        ymin=0,
        ymax=4400,
        ytick={0, 1000, ..., 4000},
        xticklabel style={yshift=2mm},
        xlabel style={yshift=0mm},
        legend style={
            at={(1.02,1.00)},anchor=north west,column sep=2pt,
            draw=black,fill=white,line width=.5pt,
            /tikz/every even column/.append style={column sep=5pt}
        },
        legend cell align={left},
        legend columns=1,
        area legend,
        clip=false,
        every axis/.append style={font=\scriptsize},
        /pgfplots/log origin=infty,
        minor tick num=1,
        minor grid style=lightgray,
        ymajorgrids,
        yminorgrids,
        },
    accfig/.belongs to family=/pgfplots/scale,
}
    \vspace{-9mm}

    \begin{tikzpicture}

    \def\yclip{4800}
    \def\ycliptop{4800}
    \def\ycliplowerlower{4200}
    \def\ycliplower{4300}
    \def\yclipupper{4500}
    \def\yclipupperupper{4600}
    \newcommand\xmid[2]{#1 + #2}
    \newcommand\xlower[3]{\xmid{#1}{#2} - #3}
    \newcommand\xupper[3]{\xmid{#1}{#2} + #3}
    
    \newcommand{\slabClipNumber}[4]{%
        \draw[fill=white,draw=white] (axis cs: \xlower{#1}{#2}{#3},\yclipupper) -- (axis cs: \xupper{#1}{#2}{#3},\yclipupperupper) -- (axis cs: \xupper{#1}{#2}{#3},\ycliplower) -- (axis cs: \xlower{#1}{#2}{#3},\ycliplowerlower) -- cycle;%
        \draw[draw=black] (axis cs: \xlower{#1}{#2}{#3},\yclipupper) -- (axis cs: \xupper{#1}{#2}{#3},\yclipupperupper);%
        \draw[draw=black] (axis cs: \xupper{#1}{#2}{#3},\ycliplower) -- (axis cs: \xlower{#1}{#2}{#3},\ycliplowerlower);%
        \node[anchor=west, rotate=90] at (axis cs: \xmid{#1}{#2},\ycliptop) {#4};%
    }
    
    \begin{axis}[accfig,
        ylabel=Search Latency (ms)]
    
    \addplot[thick,draw=vintageblack,fill=vintageblack,
    postaction={
        pattern=north east lines,
        pattern color=white,
    }]
    plot [error bars/.cd, y dir=plus,y explicit,error bar style={color=black}]
    table[x=x,y=y,y error expr=\thisrowno{3}-\thisrowno{2}] {
    x c y    y+
    1 0 254.8345557 3885.705658
    2 0 839.7534547 {\yclip}
    3 0 2099.758349 {\yclip}
    };
    \slabClipNumber{2}{-0.25}{0.05}{14117};
    \slabClipNumber{3}{-0.25}{0.05}{37431};

    \addplot[thick,draw=vintageblack,fill=vintageblack]
    plot [error bars/.cd, y dir=plus,y explicit,error bar style={color=black}]
    table[x=x,y=y,y error expr=\thisrowno{3}-\thisrowno{2}] {
    x c y    y+
    1 1 1655.994508 {\yclip}
    2 1 1643.60251  {\yclip}
    3 1 2128.250344 {\yclip}
    };
    \slabClipNumber{1}{-0.125}{0.05}{16474};
    \slabClipNumber{2}{-0.125}{0.05}{16817};
    \slabClipNumber{3}{-0.125}{0.05}{14688};
    
    \addplot[thick,draw=vintageblack!50,fill=vintageblack!50,]
    plot [error bars/.cd, y dir=plus,y explicit,error bar style={color=black}]
    table[x=x,y=y,y error expr=\thisrowno{3}-\thisrowno{2}] {
    x c y    y+
    1 2 242.0121507 618.240363
    2 2 782.4309524 1628.206049
    3 2 1943.175542 4208.20374
    };
    
    \addplot[thick,draw=vintageblack!50,fill=vintageblack!50,
    postaction={
        pattern=north east lines,
        pattern color=black,
    }]
    plot [error bars/.cd, y dir=plus,y explicit,error bar style={color=black}]
    table[x=x,y=y,y error expr=\thisrowno{3}-\thisrowno{2}] {
    x c y    y+
    1 3 {\yclip} {\yclip}
    2 3 {\yclip} {\yclip}
    3 3 {\yclip} {\yclip}
    };
    \slabClipNumber{1}{0.125}{0.05}{6781};
    \slabClipNumber{2}{0.125}{0.05}{7817};
    \slabClipNumber{3}{0.125}{0.05}{10545};
    
    
    \addplot[thick,draw=vintageorange,fill=vintageorange,]
    plot [error bars/.cd, y dir=plus,y explicit,error bar style={color=black}]
    table[x=x,y=y,y error expr=\thisrowno{3}-\thisrowno{2}] {
    x c y    y+
    1 4 175.3187697 477.434309
    2 4 414.801396  1070.681104
    3 4 1132.139502 2783.074737
    };

    \addlegendentry{Lucene}
    \addlegendentry{Elasticsearch}
    \addlegendentry{SQLite}
    \addlegendentry{HashTable}
    \addlegendentry{\system}
      
    \end{axis}
    \end{tikzpicture}
    

  \vspace{-2mm}
  \caption{End-to-end search latencies across regions of indexes.
  The \texttt{Windows} dataset was used. 
  Solid bars show average latencies; the upper error bars show 99th percentiles.
  of measured latencies. 
  }
  
  \label{fig:e2e_search_region}
  \vspace{-3mm}
\end{figure}
    

Separation of compute and storage allows \system to host each component in different physical locations, even across continents. To construct such a scenario, we allocate a cloud storage in the default \texttt{US} (multi-region US), while hosting VMs at \texttt{us-central1-c} (Iowa), \texttt{europe-west2-c} (London), and \texttt{asia-southeast1-b} (Singapore). Because we observe similar performance patterns across all datasets, \cref{fig:e2e_search_region} only presents \texttt{Windows} as a representative. It comes as no surprise that each method takes longer as its VM moves further from the data. Among those baselines with competitive latencies, Lucene is $3.3\times$ and $8.2\times$ slower and SQLite is $3.2\times$ and $8.0\times$ slower in London and Singapore. In contrast, \system achieves a milder slowdown: $2.4\times$ and $6.5\times$ ($3.3\times$ and $6.8\times$ across all datasets, not shown). 
Elasticsearch and HashTable, on the other hand, are consistently slower than others, across different regions:
Elasticsearch spends much time in mounting its searchable snapshots; HashTable reads too many false positives.

\paragraph{Latency Breakdown}
\label{sec:experiment:query_breakdown}

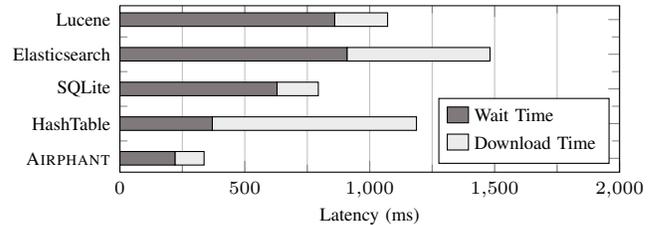
\begin{figure}[t]

\pgfplotstableread{
method low high
lucene 859.6774193548387 211.93548387096777
elasticsearch 909.0 572.5
sqlite 629.032258064516 165.1612903225806
hashtable 370.0 817.0967741935484
iou0 220.33333333333331 116.66666666666666
}\dataTransferTable

\centering
\begin{tikzpicture}
    \begin{axis}[
        xbar stacked,
        scale only axis,
        bar width=1.8mm,
        width=0.75\linewidth,
        height=0.25\linewidth,
        xmin=0,
        xmax=2000,
        xtick={0, 500, ..., 2000},
        minor tick num=1,
        xmajorgrids,
        xminorgrids,
        y dir=reverse,
        ytick=data,
        xlabel near ticks,
        ylabel near ticks,
        xlabel=Latency (ms),
        xlabel shift=-1mm,
        yticklabels={Lucene,Elasticsearch,SQLite,HashTable,\system},
        every axis/.append style={font=\scriptsize},
        legend style={
            font=\scriptsize,
            at={(0.98,0.05)}, anchor=south east,
            draw=black,
            fill=white
        },
        legend cell align={left},
    ]
        \addplot [
            fill=vintageblack!70
        ] table [x=low, meta=method,y expr=\coordindex] {\dataTransferTable};
        \addplot [
            fill=vintageblack!10,
        ] table [x=high, meta=method,y expr=\coordindex] {\dataTransferTable};
        \legend{Wait Time,Download Time}
    \end{axis}
\end{tikzpicture}
\vspace{-2mm}

\caption{Search latency breakdown in terms of network communication on the \texttt{Spark} dataset. The total search latency is equal to the summation of wait time and download time.}
\vspace{-2mm}

\label{fig:data_transfer}
\end{figure}





To dig deeper, we study data movement patterns from the perspective of network communication via TCP packets captured by \texttt{tcpdump}\footnote{\url{https://www.tcpdump.org}}. We sample 32 queries from each method, capture packets on appropriate ports. 
We then compute two metrics per query: waiting time and download time.\footnote{The former includes the time when the bandwidth usage is less than $10$ kB/s while the latter includes the rest where the traffic is high.} We calculate bandwidth usage based on moving averages of the size of packets with window size of $10$ ms. \cref{fig:data_transfer} depicts the result of this procedure. Note that the latency is slower than those in \cref{fig:e2e_search} partly due to \texttt{tcpdump} and a smaller number of query samples. There are two opposite patterns among baselines. 
First, as we have mentioned, HashTable's latency dominantly consists of time to download false-positive documents. 
Secondly and conversely, Lucene and SQLite tend to wait for their dependent reads, i.e. B-tree or skip list traversal. 
\system minimizes such dependency thanks to the independence among IoU Sketch's layers, reducing wait time as a result. \system also improves download time because of its parallel I/O reads which utilizes more network bandwidth. \system spends 220ms waiting and 117ms downloading on average where its breakdown distribution concentrates around.

\subsection{Cost Comparison}
\label{sec:experiment:cost_compare}

    This experiment compares costs between the two architecture paradigms: coupled (compute and storage adjacency) and decoupled (compute and storage separation). To this end, we consider the \textit{peak-trough} workload that elucidates the trade-off between the two paradigms. The peak-trough workload is motivated by the periodically varying utilization commonly found in web services. As the name suggests, it consists of two parts: peak and trough. Peak refers to the time where the workload is higher, whereas trough's workload is lower. Because our following cost formulae are linear on the amount of time, we can identify a peak-trough instance with $(A, a, \tau)$ where peak and trough cover $\tau$ and $1 - \tau$ fraction of time with workload $A$ and $a$ ops/s respectively.
    
    Furthermore, we also consider the total number of corpora. For simplicity, we consider a collection of corpora whose document-word pairs are distributed similarly to \texttt{Windows}. We denote $S$ as the total size of original data in bytes. Estimated from their storage usage in \texttt{Windows}, \system uses $1.008 \times S$ bytes, while Elasticsearch has a better compression rate and only uses $0.3316 \times S$ bytes of storage.
    
    Since \system can scale up/down easily as per workloads, \system costs proportionally to the workload over time: $O(a_p t + a_t (1 - t))$. From previous performance results, \system operates at $175$ ms/op or $5.71$ ops/s. 
    Deploying \system on \texttt{e2-small} VM costs \$13.23/month,
    and storing its index on GCP Cloud Storage costs \$0.02/GB/month.
    
    
    In contrast, Elasticsearch cannot automatically scale down without rebalancing its index over remaining servers; as a result, it requires the resource to handle the peak workload at all times, amounting to $O(a_p)$ cost. Nonetheless, we optimistically assume that Elasticsearch's sharding and load balancing are perfect such that its throughput scales linearly with the number of servers without shard replication. We deploy Elasticsearch on \texttt{e2-medium} VM (\$26.46/month) and store its index on local disk (\$0.2/GB/month). This achieves 6.49 ms/op and is the most cost-efficient among the options we explored.
    

\begin{figure}[t]
\vspace{-2mm}
    \centering
    \begin{tikzpicture}
    
    \def\Xmin{0.0}
    \def\Xmax{100.0}
    \def\DataSize{28012696901}              
    \def\Tarray{16, 8, 4, 2, 1}
    \def\Tmin{1}
    \def\Tmax{16}
    
    \begin{axis}[
        scale only axis,
        width=0.6\linewidth,
        height=0.26\linewidth,
        xlabel=Fraction of Peak Time $\tau$,
        ylabel=Relative Cost $C_E/C_A$,
        xlabel near ticks,
        ylabel near ticks,
        every axis/.append style={font=\scriptsize},
        xmin=\Xmin,
        xmax=\Xmax,
        domain=\Xmin:\Xmax,
        ymin=0.25,
        ymax=4,
        ymode=log,
        colormap/viridis,
        xtick={0, 20, 40, 60, 80, 100},
        xticklabels={0\%, 20\%, 40\%, 60\%, 80\%, 100\%},
        ytick={0.25, 0.5, 1, 2, 4},
        yticklabels={1/4, 1/2, 1, 2, 4},
        yticklabel style={
            align=right,
            inner sep=0pt,
            xshift=-0.1cm
        },
        ymajorgrids,
        minor grid style=lightgray,
        legend cell align={left},
        legend style={
            font=\scriptsize,
            at={(1.05,1.0)}, anchor=north west,
            draw=black,
            fill=white
        },
        colormap/jet,
        colorT/.style={%
            /utils/exec={%
                \PgfmathparseFPU{#1}%
                \pgfplotscolormapdefinemappedcolor{\pgfmathresult}%
            },%
            color=mapped color%
        }], 
    ]
    
    \addlegendimage{empty legend}
    \addlegendentry{\hspace{-0.4cm}Size $N$}
    
    \foreach \T in \Tarray {
        \edef\temp{\noexpand\addplot 
        [
            thick,
            domain=\Xmin:\Xmax,
            samples=50, 
            colorT={(ln(\T)-ln(\Tmin))*(1000/(ln(\Tmax)-ln(\Tmin)))},
        ]{ 
            ((26.46 * 6.49 * 0.15408) + 0.0000000002 * (9287623285 / \DataSize) * \T * 1000000000000)
            /
            (13.23 * 175.3187697 * ((1 - x/100)*0.007704 + (x/100)*0.15408) + 0.00000000002 * (28245668027 / \DataSize) * \T * 1000000000000)
        };
        \noexpand\addlegendentryexpanded{$\T$ TB}
        }
        \temp
    }
    
    \addplot[thick]{1.0};
        
    \end{axis}
    \end{tikzpicture}
    
    \vspace{-2mm}

    \caption{Relative cost between local Elasticsearch and cloud-stored \system in terms of fraction of peak time $\tau$ and size of indexed data $N$ (increasing from bottom to top lines). Peak and trough workloads are fixed to $A = 154.08$ op/s (throughput of a single Elasticsearch server) and $a = A / 20 = 7.704$ op/s.}
    \label{fig:cost_compare}
  
  \vspace{-2mm}
\end{figure}
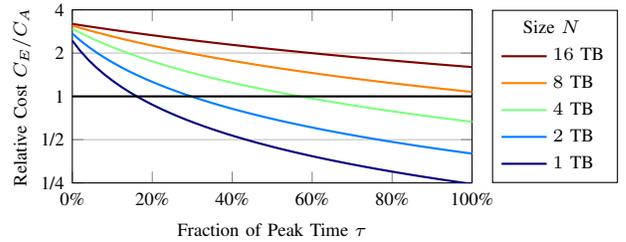

\input{figures/exp_iou_tradeoff}
    
\cref{fig:cost_compare} shows a slice of the relative cost $C_E / C_A$ at fixed $A$ and $a$. Overall, \system costs less if the total corpus size $N$ is larger and/or the fraction of peak time $\tau$ is smaller. In fact, we would asymptotically save $\lim_{N \rightarrow \infty}
    C_E/C_A
    \approx 3.29$ times of the cost of coupled Elasticsearch. Also, focusing on the VM cost, \system's cost would be 
    $A / (13.48 a)$ times over Elasticsearch's. A relatively higher peak workload $A > 13.48 a$ implies that a relatively lower cost to deploy \system, and vice versa. These trends align with the intuition: the decouple paradigm is more cost-efficient than coupled one when 1) the data is large, 2) the workload is concentrated over peak times, and/or 3) the peak's workload is high compared to trough's.

\subsection{IoU Sketch Structure}
\label{sec:experiment:trafeoff-iou}




We use \texttt{HDFS} to explore $B \in \{50k, 100k, 200k, 400k \}$ around the configuration used in \cref{sec:experiment:e2e} and a wide range of $L \in \{1, 2, 4, 6, 8, \dots, 16 \}$. Recall that $L = 1$ corresponds to the na\"ive hash table index structure and note that our optimizer selects $L^{*} = 2$.
\cref{fig:tradeoff:false-positives} confirms our analysis which predicts a rapid error reduction over $L$ given a sufficient $B$; the observed average numbers of false positives are enormous at $L=1$, become less than $1$ at $L=2$, and are exactly zero after $L=4$. Such high numbers of false positives in $L=1$ result in a higher search latency (\cref{fig:tradeoff:search-latency}) due to heavy filtering demands. The rest of the search latency trend correlates closely with the term lookup latency (\cref{fig:tradeoff:lookup-latency}); when $L$ increases, so does the lookup latency. Even though \system fetches superposts in layers separately, a bandwidth contention drives up the lookup latency and search latency as a result.



\section{Related Work}

\subsection{Data Systems in the Cloud}

Cloud services provide a convenient interface to share computing resources. As one of its most important impact, the cost of scaling up or down is now low enough that elasticity, or ability to scale to meet the demand, is a desirable feature of a system. Some work achieves elasticity through separation of compute and storage. Particularly, applications that are functional or can be partitioned into smaller functional modules enjoy elasticity by FaaS \cite{Schleier2019}. Apart from web serving and file processing, \texttt{gg} framework \cite{Fouladi2019} empowers FaaS execution for commonly local applications, for example, video encoding, object recognition, unit testing. Towards FaaS-based query engine, Starling \cite{Perron2020} partitions a query into fine-grained tasks on operation level such as table scanning, joining, and shuffling. \system partitions a functional searching into \system Searcher, and persists index data structure stored on cloud storage. In this way, \system can also benefit from FaaS elasticity.




\subsection{Index Data Structures}
\label{sec:indexing}

\paragraph{Sketches}

\cite{Cormode2011} lays out numerous sketch techniques. Most notably similar sketches to IoU Sketch sketches are count-min sketch \cite{Cormode2005} and Bloom filter \cite{Bloom1970}. IoU Sketch and count-min sketch both maintain a multi-layer hash table and answer queries using aggregations. Bloom filter's accuracy formulation is similar to IoU Sketch's counterpart; however, IoU Sketch focuses on maximizing its retrieval performance using its accuracy term as one of the constraints while Bloom filter directly maximizes its accuracy. Moreover, Bloom filter has a closed-form optimal configuration, but IoU Sketch's can contain multiple local optima.

\paragraph{Hash Tables}

Open addressing probes for available entries to insert colliding values accordingly to a mechanism, for example,
Hopscotch \cite{Herlihy2008}, SmartCuckoo \cite{Sun2019}, and MinCounter \cite{Sun2015}. Open addressing performs well when load factor is low \cite{Richter2016}, which is infeasible with a memory constraint.

Unlike B-trees, hash table is popular in distributed peer-to-peer systems, especially for point-query indexing. Each solution differs from others based on its query routing \cite{Balakrishnan2003}, such as tree-like (Pastry \cite{Rowstron2001}, Tapestry \cite{Hildrum2002}, Kademlia \cite{Maymounkov2002}) or skip-list-like (Chord \cite{Stoica2001}, Koorde \cite{Kaashoek2003}). A keyword searching application can rely on these systems to store and retrieve postings lists; however, it would require long-running servers equipped with storage, as opposed to our serverless system.

\paragraph{Hierarchical Indexes}

    B-tree \cite{Bayer1972} is widely used due to its efficiency, self-balancing, and cheap reorganization \cite{Comer1979}. Many B-tree variants and techniques optimize B-tree across memory levels and storage mediums, e.g. CPU cache \cite{Graefe2001}, SRAM cache \cite{Chen2001}, disk \cite{Chen2002}, SSD \cite{Jin2011}, NVMe \cite{Wang2020}, and distributed systems \cite{Wu2010, Zhou2014, Bin2014}.
    
    Skip list is another index structure based on linked list with skip connections.
    Like B-trees, many skip lists are optimized to various settings, such as multi-core \cite{Dick2017}, cache-sensitive for range queries \cite{Sprenger2017}, non-uniform access \cite{Daly2018}, and distributed nodes \cite{He2018}. Skip list can be combined with other data structures as well: hash table \cite{Stoica2001, Kaashoek2003}, search tree \cite{Zhang2018}, multi-dimensional octree~\cite{He2016}, and compressed perfect skip list for inverted index \cite{Boldi2005}.
    
    However, these hierarchical index structure can be inefficient in our cloud setting when corpus size is large. Although \cite{Graefe2010} suggests increasing the effective fan-out, a server with limited memory still requires multiple sequential round-trips. A workaround is to increase the server's memory capacity, thereby incurring extra costs. \system achieves single round-trip per query even on a memory-constrained node.

\subsection{Approximation Techniques for Data Systems}

Besides sketching (\cref{sec:indexing}), approximation and 
learning-based techniques have been proposed to speed up 
large-scale data analytics~\cite{park2017database,kraska2019sagedb,park2018verdictdb,he2018demonstration,bater2020saqe,park2019blinkml,park2016visualization}, 
selectivity estimation~\cite{yang13deep,marcus12neo,park2020quicksel,kipf2018learned},
etc.~and 
to automate database tuning~\cite{zhang2019end,van2017automatic,zhu2017bestconfig}, etc.
In contrast, \system develops a novel statistical inverted index for document stores.

\section{Conclusion}

This work presents \system, a search engine developed to achieve low end-to-end query latencies
under the separation of compute and storage.
While \system keeps almost the entire data on cloud storage (both inverted index and original documents),
it delivers search results within 300 milliseconds.
This is a significant improvement compared to placing existing search engines directly on top of cloud storage.
At its core, \system relies on our new statistical inverted indexing technique called IoU Sketch.
Unlike many other existing indexes,
IoU Sketch does not require sequential back-to-back communications with storage devices.
With IoU Sketch, we can make a single batch of concurrent communications, which significantly lowers the end-to-end wall clock time
in obtaining relevant documents (i.e., the documents containing search keywords).
We plan to open-source this technology (and \system itself) after enhancing its index building component
(for higher scalability)
and also making it more extensible for diverse public cloud environments.

\section{Acknowledgement}

This work is supported in part by Microsoft Azure and Google Cloud Platform.

\bibliographystyle{IEEEtran}
\bibliography{bib/airphant}
%



\clearpage
\appendices

\section{Additional Visualizations}
\label{sec:apdx:viz}

Here we include additional visualizations that we trim in the main sections due to space limitation.

\subsection{Individual Latency Breakdown}
\label{sec:apdx:viz:individual_breakdown}

\cref{sec:experiment:query_breakdown} summarizes search latency breakdown averages. \cref{fig:data_transfer_scatter} instead visualizes latency breakdowns for individual search queries across systems. Most visibly perhaps is the fact that search latencies vary across queries due to many factors such as the result size, network variability, or buffer pool state. Nevertheless, the underlying search engine significantly dictates the breakdown. Overall, \system outperforms other systems by minimizing both wait time and download time simultaneously.

On a closer look, this scatter plot reveals two extreme access patterns. First, wait-heavy systems (e.g. Lucene) spend most of the search waiting for network responses. The most probable culprit is their dependent sequential reads, that is, reads whose locations depend on decisions in preceding reads. For example, skip list traversal requires the current node to find the next node to skip to; therefore, to know which block to read next, the skip list needs to complete reading the current node first. Second, download-heavy systems (e.g. HashTable) spends most of the search downloading their data. These systems minimize their number of roundtrips but consume more data. In particular, HashTable (i.e. IoU Sketch where $L=1$) contains a large number of false positives due to its limited filtering power. As a result, it unnecessarily reads a lot more false-positive documents but avoids roundtrips thanks to prefetching.

\begin{figure}[t]

\pgfplotstableread{
method low high
lucene 1.170 0.300
lucene 1.490 0.380
lucene 0.870 0.190
lucene 1.450 0.290
lucene 0.190 0.070
lucene 1.450 0.450
lucene 0.760 0.190
lucene 1.180 0.280
lucene 0.730 0.090
lucene 0.300 0.070
lucene 0.500 0.110
lucene 0.520 0.050
lucene 0.220 0.080
lucene 2.790 0.770
lucene 0.270 0.090
lucene 0.520 0.170
lucene 1.440 0.380
lucene 0.530 0.120
lucene 2.390 0.510
lucene 0.410 0.140
lucene 0.240 0.050
lucene 0.270 0.050
lucene 1.110 0.340
lucene 0.500 0.110
lucene 0.780 0.210
lucene 0.740 0.120
lucene 0.260 0.070
lucene 0.990 0.270
lucene 0.790 0.220
lucene 0.400 0.110
lucene 1.390 0.290
elasticsearch 3.280 1.440
elasticsearch 1.420 0.490
elasticsearch 0.410 0.150
elasticsearch 0.420 0.360
elasticsearch 0.610 0.500
elasticsearch 0.120 0.030
elasticsearch 0.380 0.290
elasticsearch 0.930 0.370
elasticsearch 0.280 0.320
elasticsearch 0.360 0.080
elasticsearch 1.990 1.910
elasticsearch 0.460 0.180
elasticsearch 0.570 0.580
elasticsearch 4.480 2.640
elasticsearch 0.150 0.120
elasticsearch 0.230 0.070
elasticsearch 0.210 0.310
elasticsearch 0.580 0.500
elasticsearch 0.540 0.500
elasticsearch 0.760 0.610
sqlite 0.780 0.150
sqlite 0.690 0.190
sqlite 0.710 0.150
sqlite 0.600 0.180
sqlite 0.490 0.120
sqlite 0.640 0.170
sqlite 0.560 0.130
sqlite 0.710 0.130
sqlite 0.850 0.160
sqlite 0.650 0.160
sqlite 0.770 0.150
sqlite 0.760 0.180
sqlite 0.570 0.160
sqlite 0.420 0.240
sqlite 0.610 0.130
sqlite 0.610 0.160
sqlite 0.460 0.190
sqlite 0.540 0.160
sqlite 1.210 0.430
sqlite 0.790 0.140
sqlite 0.590 0.160
sqlite 0.390 0.110
sqlite 0.540 0.160
sqlite 0.670 0.170
sqlite 0.520 0.100
sqlite 0.470 0.130
sqlite 0.580 0.160
sqlite 0.820 0.170
sqlite 0.500 0.140
sqlite 0.590 0.130
sqlite 0.410 0.210
hashtable 0.230 0.790
hashtable 0.310 0.470
hashtable 0.340 0.500
hashtable 0.590 0.740
hashtable 0.210 0.860
hashtable 0.360 1.350
hashtable 0.210 0.380
hashtable 0.380 0.660
hashtable 0.370 0.380
hashtable 0.600 1.350
hashtable 0.300 0.520
hashtable 2.060 2.710
hashtable 0.260 0.800
hashtable 0.220 0.310
hashtable 0.150 0.510
hashtable 0.250 0.410
hashtable 0.290 0.620
hashtable 0.260 1.040
hashtable 0.100 0.360
hashtable 0.440 1.470
hashtable 0.760 0.420
hashtable 0.170 0.670
hashtable 0.320 1.510
hashtable 0.280 1.660
hashtable 0.300 0.840
hashtable 0.210 0.490
hashtable 0.170 0.500
hashtable 0.520 0.760
hashtable 0.280 0.560
hashtable 0.220 1.300
hashtable 0.300 0.400
iou 0.130 0.120
iou 0.180 0.140
iou 0.180 0.110
iou 0.170 0.070
iou 0.170 0.130
iou 0.210 0.070
iou 0.140 0.070
iou 0.240 0.130
iou 0.180 0.080
iou 0.170 0.080
iou 0.190 0.130
iou 0.210 0.130
iou 0.310 0.290
iou 0.420 0.080
iou 0.160 0.090
iou 0.330 0.130
iou 0.200 0.130
iou 0.440 0.440
iou 0.210 0.060
iou 0.190 0.060
iou 0.260 0.040
iou 0.270 0.080
iou 0.280 0.090
iou 0.180 0.120
iou 0.180 0.090
iou 0.170 0.070
iou 0.230 0.130
iou 0.170 0.090
iou 0.180 0.080
iou 0.260 0.170
}\dataTransferTable

\centering
\begin{tikzpicture}
    \begin{axis}[
        scale only axis,
        width=0.8\linewidth,
        height=0.48\linewidth,
        xmin=0,
        xmax=5,
        ymin=0,
        ymax=3,
        xtick={0, 0.5, ..., 5},
        ytick={0, 0.5, ..., 3},
        xlabel near ticks,
        ylabel near ticks,
        xlabel=Wait Time (s),
        ylabel=Download Time (s),
        every axis/.append style={font=\scriptsize},
        legend pos=south east,
        legend style={
            font=\scriptsize,
            draw=black,
            fill=white
        },
        legend cell align={left},
        legend columns=1
    ]
        \addplot[
            scatter,
            only marks,
            scatter src=explicit symbolic,
            scatter/classes={
                lucene={mark=triangle*,green},
                elasticsearch={mark=diamond,cyan},
                sqlite={mark=x,blue!75},
                hashtable={mark=square,black!50},
                iou={mark=o,vintageorange}
            },
        ]
        table[x=low,y=high,meta=method]{\dataTransferTable};
        \legend{Lucene,Elasticsearch,SQLite,HashTable,Airphant}
    \end{axis}
\end{tikzpicture}
\vspace{-2mm}

\caption{Individual search latency breakdown in terms of network communication on the \texttt{Spark} dataset. Each scatter mark represents how long each individual search spends waiting and downloading. The total search latency is equal to the summation of wait time and download time. Faster searches reside in the lower left corner.}
\vspace{-2mm}

\label{fig:data_transfer_scatter}
\end{figure}
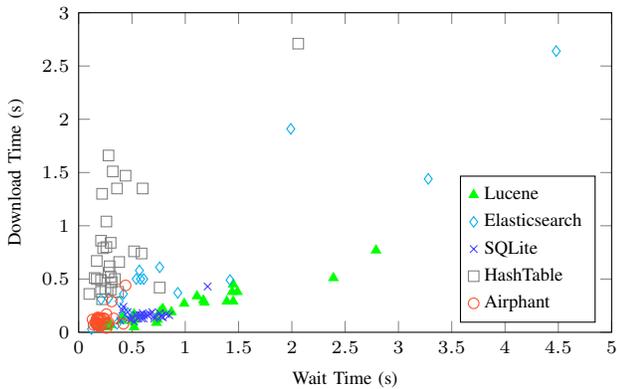





\subsection{Full Cross-Region Results}
\label{sec:apdx:viz:full_cross_region}

\cref{sec:experiment:cross_region} shows search latency trends as we move further away from the cloud storage; however, it only selects those measurements based on \texttt{Windows} dataset as the representative. \cref{fig:e2e_search_eu} and \cref{fig:e2e_search_asia} fully display the measurements across all 7 datasets.

\begin{figure*}[t]
  \centering

\pgfplotsset{
    accfig/.style={
        width=0.85\linewidth,
        height=42mm,
        ybar,
        bar width=1.2mm,
        xmin=0.5,
        xmax=7.5,
        xtick={1, 2, ..., 7},
        xticklabels={\texttt{diag(8,8,0)}, \texttt{uniform(8,8,1)}, \texttt{zipf(8,8,1)}, \texttt{Cranfield}, \texttt{HDFS}, \texttt{Windows}, \texttt{Spark}},
        ylabel near ticks,
        ylabel style={align=center},
        ymax=10000000,
        ytick={1, 10, 100, 1000, 10000, 100000, 1000000, 10000000},
        xticklabel style={yshift=2mm},
        xlabel style={yshift=0mm},
        legend style={
            at={(1.02,1.0)},anchor=north west,column sep=2pt,
            draw=black,fill=white,line width=.5pt,
            /tikz/every even column/.append style={column sep=5pt}
        },
        legend cell align={left},
        legend columns=1,
        area legend,
        clip=false,
        every axis/.append style={font=\scriptsize},
        /pgfplots/ymode=log,
        /pgfplots/log origin=infty,
        minor grid style=lightgray,
        ymajorgrids,
        },
    accfig/.belongs to family=/pgfplots/scale,
}

    \begin{tikzpicture}
    \begin{axis}[accfig,
        ylabel=Search Latency (ms)]
    
    \addplot[thick,draw=vintageblack,fill=vintageblack,
    postaction={
        pattern=north east lines,
        pattern color=white,
    }]
    plot [error bars/.cd, y dir=plus,y explicit,error bar style={color=black}]
    table[x=x,y=y,y error expr=\thisrowno{3}-\thisrowno{2}] {
    x c y    y+
    1 0 572.8089254 951.222407
    2 0 7957.058236 13205.43993
    3 0 4168.483895 13606.32159
    4 0 3.241360334 111.46959
    5 0 875.1827639 2820.877001
    6 0 839.7534547 14117.8497
    7 0 1527.858987 8065.231209
    };
    
    \addplot[thick,draw=vintageblack,fill=vintageblack]
    plot [error bars/.cd, y dir=plus,y explicit,error bar style={color=black}]
    table[x=x,y=y,y error expr=\thisrowno{3}-\thisrowno{2}] {
    x c y    y+
    1 1 233.7850197 1366.522266
    2 1 33832.12281 46119.97659
    3 1 3159.419616 16110.53049
    4 1 18.86058253 52.511735
    5 1 212.843206 718.343901
    6 1 1643.60251  16817.4931
    7 1 245.6759188 1426.898492
    };
    
    \addplot[thick,draw=vintageblack!50,fill=vintageblack!50,]
    plot [error bars/.cd, y dir=plus,y explicit,error bar style={color=black}]
    table[x=x,y=y,y error expr=\thisrowno{3}-\thisrowno{2}] {
    x c y    y+
    1 2 932.7794291 1545.059984
    2 2 1305.418526 2643.745226
    3 2 1081.958393 1960.346511
    4 2 47.79461557 483.367473
    5 2 624.0733924 1381.617071
    6 2 782.4309524 1628.206049
    7 2 839.3777734 1908.36181
    };
    
    \addplot[thick,draw=vintageblack!50,fill=vintageblack!50,
    postaction={
        pattern=north east lines,
        pattern color=black,
    }]
    plot [error bars/.cd, y dir=plus,y explicit,error bar style={color=black}]
    table[x=x,y=y,y error expr=\thisrowno{3}-\thisrowno{2}] {
    x c y    y+
    1 3 3255.952346 4205.011176
    2 3 97557.41994 2612286.755
    3 3 24311.92782 41411.92118
    4 3 40.43526768 352.028549
    5 3 1052.742456 2809.276402
    6 3 7817.122712 50454.55853
    7 3 1718.490775 6518.456685
    };
    
    \addplot[thick,draw=vintageorange,fill=vintageorange,]
    plot [error bars/.cd, y dir=plus,y explicit,error bar style={color=black}]
    table[x=x,y=y,y error expr=\thisrowno{3}-\thisrowno{2}] {
    x c y    y+
    1 4 376.8702573 710.367971
    2 4 651.4843373 1602.690007
    3 4 479.9272433 1321.765812
    4 4 37.4948442  363.950979
    5 4 289.1142445 643.895592
    6 4 414.801396  1070.681104
    7 4 326.5863024 1002.033521
    };

    \addlegendentry{Lucene}
    \addlegendentry{Elasticsearch}
    \addlegendentry{SQLite}
    \addlegendentry{HashTable}
    \addlegendentry{\system}
      
    \end{axis}
    \end{tikzpicture}
    

  \vspace{-2mm}
  \caption{End-to-end search latencies of indexes on different datasets from \textit{Europe (London)}. Solid bars show average latencies while the upper error bars show 99th percentiles of measured latencies. Notice the logarithmic scale on the y-axis.}
  
  \label{fig:e2e_search_eu}
  \vspace{-2mm}
\end{figure*}
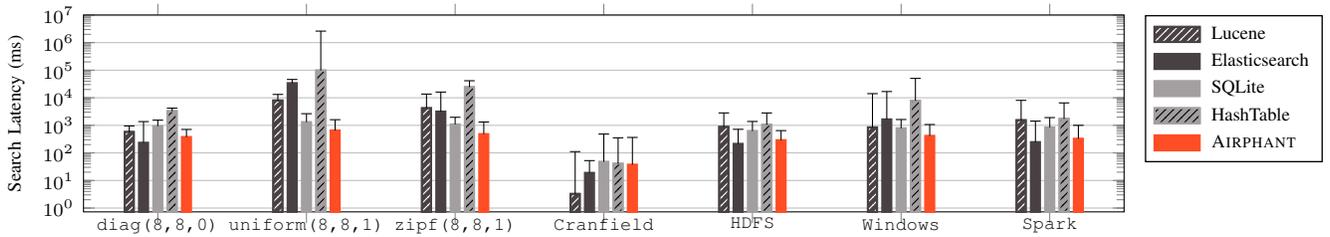

\begin{figure*}[t]
  \centering

\pgfplotsset{
    accfig/.style={
        width=0.85\linewidth,
        height=42mm,
        ybar,
        bar width=1.2mm,
        xmin=0.5,
        xmax=7.5,
        xtick={1, 2, ..., 7},
        xticklabels={\texttt{diag(8,8,0)}, \texttt{uniform(8,8,1)}, \texttt{zipf(8,8,1)}, \texttt{Cranfield}, \texttt{HDFS}, \texttt{Windows}, \texttt{Spark}},
        ylabel near ticks,
        ylabel style={align=center},
        ymax=10000000,
        ytick={1, 10, 100, 1000, 10000, 100000, 1000000, 10000000},
        xticklabel style={yshift=2mm},
        xlabel style={yshift=0mm},
        legend style={
            at={(1.02,1.0)},anchor=north west,column sep=2pt,
            draw=black,fill=white,line width=.5pt,
            /tikz/every even column/.append style={column sep=5pt}
        },
        legend cell align={left},
        legend columns=1,
        area legend,
        clip=false,
        every axis/.append style={font=\scriptsize},
        /pgfplots/ymode=log,
        /pgfplots/log origin=infty,
        minor grid style=lightgray,
        ymajorgrids,
        },
    accfig/.belongs to family=/pgfplots/scale,
}

    \begin{tikzpicture}
    \begin{axis}[accfig,
        ylabel=Search Latency (ms)]
    
    \addplot[thick,draw=vintageblack,fill=vintageblack,
    postaction={
        pattern=north east lines,
        pattern color=white,
    }]
    plot [error bars/.cd, y dir=plus,y explicit,error bar style={color=black}]
    table[x=x,y=y,y error expr=\thisrowno{3}-\thisrowno{2}] {
    x c y    y+
    1 0 1321.891674 2685.048415
    2 0 16061.12763 27265.46487
    3 0 9531.471045 29269.46282
    4 0 6.796703632 239.822189
    5 0 2223.104264 7550.617328
    6 0 2099.758349 37431.87928
    7 0 3258.418332 16604.84533
    };
    
    \addplot[thick,draw=vintageblack,fill=vintageblack]
    plot [error bars/.cd, y dir=plus,y explicit,error bar style={color=black}]
    table[x=x,y=y,y error expr=\thisrowno{3}-\thisrowno{2}] {
    x c y    y+
    1 1 546.5974926 4683.444322
    2 1 34067.92064 46758.9028
    3 1 3530.732237 18530.97371
    4 1 23.48080481 48.077612
    5 1 469.6573491 1904.824206
    6 1 2128.250344 14688.09089
    7 1 557.0473308 4586.264348
    };
    
    \addplot[thick,draw=vintageblack!50,fill=vintageblack!50,]
    plot [error bars/.cd, y dir=plus,y explicit,error bar style={color=black}]
    table[x=x,y=y,y error expr=\thisrowno{3}-\thisrowno{2}] {
    x c y    y+
    1 2 2124.899295 3385.997038
    2 2 3346.6223 5755.281672
    3 2 2448.838881 4614.796388
    4 2 87.70456329 930.841873
    5 2 1499.285923 3238.860338
    6 2 1943.175542 4208.20374
    7 2 2063.905187 5657.226302
    };
    
    \addplot[thick,draw=vintageblack!50,fill=vintageblack!50,
    postaction={
        pattern=north east lines,
        pattern color=black,
    }]
    plot [error bars/.cd, y dir=plus,y explicit,error bar style={color=black}]
    table[x=x,y=y,y error expr=\thisrowno{3}-\thisrowno{2}] {
    x c y    y+
    1 3 9462.451996 9747.283055
    2 3 119265.1996 907486.4017
    3 3 40645.01975 103011.6883
    4 3 78.9855411  712.792245
    5 3 2558.572518 5989.698391
    6 3 10545.37325 115839.6961
    7 3 4721.576236 16153.8763
    };
    
    \addplot[thick,draw=vintageorange,fill=vintageorange,]
    plot [error bars/.cd, y dir=plus,y explicit,error bar style={color=black}]
    table[x=x,y=y,y error expr=\thisrowno{3}-\thisrowno{2}] {
    x c y    y+
    1 4 730.7731417 1447.299905
    2 4 1695.563014 3107.065969
    3 4 1301.547099 3324.336061
    4 4 73.33215171 708.368106
    5 4 773.4802533 1459.609839
    6 4 1132.139502 2783.074737
    7 4 956.1917029 2479.011535
    };

    \addlegendentry{Lucene}
    \addlegendentry{Elasticsearch}
    \addlegendentry{SQLite}
    \addlegendentry{HashTable}
    \addlegendentry{\system}
      
    \end{axis}
    \end{tikzpicture}
    

  \vspace{-2mm}
  \caption{End-to-end search latencies of indexes on different datasets from \textit{Asia (Singapore)}. Solid bars show average latencies while the upper error bars show 99th percentiles of measured latencies.  Notice the logarithmic scale on the y-axis.}
  
  \label{fig:e2e_search_asia}
  \vspace{-2mm}
\end{figure*}
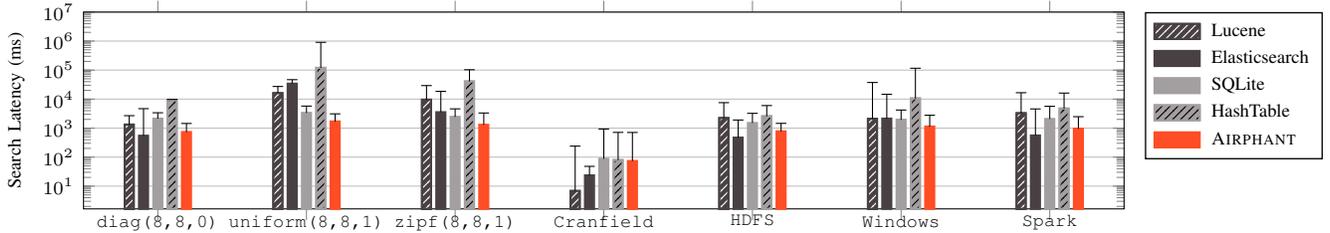

\section{Additional Results}
\label{sec:apdx:results}

In contrary to the main results, we collect the following results from March to April 2021. Some measurements might disagree with those in earlier sections.

\subsection{Term Index Lookup Performance}
\label{sec:apdx:results:term-lookup}

How does \system achieve faster latencies? Recall that \system and SQLite share the same document retrieval routine, so the difference in their latencies comes from the differences in both term index lookup operations and amounts of retrieved documents. Although \system would need to retrieve more documents to take into account the false positives, its lookup speed gain justifies additional costs. \cref{fig:e2e_lookup} reiterates the importance of \system's single-round-trip lookup operation. It evidently outperforms SQLite's cached B-tree traversal both on average and at tail latency. In the best case, \system is upto $2.79\times$ faster on average and $2.81\times$ faster at 99th percentile of term index lookup than SQLite.

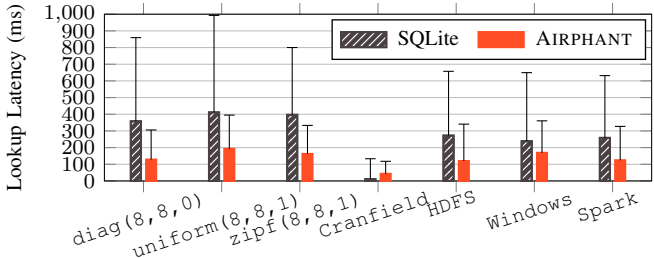
\begin{figure}[t]
  \centering

    \begin{tikzpicture}
    \begin{axis}[
        width=1.0\linewidth,
        height=38mm,
        ybar,
        bar width=0.015\linewidth,
        xmin=0.5,
        xmax=7.5,
        xtick={1, 2, ..., 7},
        xticklabels={\texttt{diag(8,8,0)}, \texttt{uniform(8,8,1)}, \texttt{zipf(8,8,1)}, \texttt{Cranfield}, \texttt{HDFS}, \texttt{Windows}, \texttt{Spark}},
        xticklabel style={rotate=18},
        ylabel near ticks,
        ylabel style={align=center},
        ymin=0,
        ymax=1000,
        ytick={0, 100, ..., 1000},
        xticklabel style={yshift=2mm},
        xlabel style={yshift=0mm},
        legend style={
            at={(0.99,0.97)},anchor=north east,column sep=2pt,
            draw=black,fill=white,line width=.5pt,
            /tikz/every even column/.append style={column sep=5pt}
        },
        legend cell align={left},
        legend columns=3,
        area legend,
        clip=false,
        every axis/.append style={font=\footnotesize},
        minor grid style=lightgray,
        ymajorgrids,
        ylabel=Lookup Latency (ms)]
    
    \addplot[thick,draw=vintageblack,fill=vintageblack,
    postaction={
        pattern=north east lines,
        pattern color=white,
    }]
    plot [error bars/.cd, y dir=plus,y explicit,error bar style={color=black}]
    table[x=x,y=y,y error expr=\thisrowno{3}-\thisrowno{2}] {
    x c y    y+
    1 0 357.400546  859.454353
    2 0 410.618882  992.836816
    3 0 396.021877  800.094758
    4 0 10.795443 133.193405
    5 0 271.673016  657.487173
    6 0 237.781243  649.503443
    7 0 257.986436  631.786886
    };
    
    
    \addplot[thick,draw=vintageorange,fill=vintageorange,]
    plot [error bars/.cd, y dir=plus,y explicit,error bar style={color=black}]
    table[x=x,y=y,y error expr=\thisrowno{3}-\thisrowno{2}] {
    x c y    y+
    1 3 127.948799  305.792355
    2 3 192.273769  394.815169
    3 3 161.941161  333.084021
    4 3 42.577726 117.69961
    5 3 119.165405  340.66282
    6 3 168.5515  360.545509
    7 3 123.973529  327.130572
    };
    
    
    
    \addlegendentry{SQLite}
    \addlegendentry{\system}
      
    \end{axis}
    \end{tikzpicture}
    
    \vspace{-2mm}

  \caption{Term index lookup latencies of indexes on different datasets. Solid bars show average latencies while the upper error bars show 99th percentiles of measured latencies.}
  
  \label{fig:e2e_lookup}
\end{figure}

\subsection{Scalability with Corpus Size}
\label{sec:apdx:results:scalability}

We conduct a scalability experiment where we vary the size of synthetic datasets. In particular, we generate $17$ datasets from \texttt{diag}, \texttt{unif}, and \texttt{zipf} by varying both numbers of documents $n_d$ and distinct words $n_w$ ranging from $10^3$ to $10^8$. \cref{fig:scalability} summarizes the measurement results of search latency and the index storage usage on \texttt{zipf} datasets. Results from \texttt{diag} and \texttt{unif} datasets possess the same pattern. It confirms that when the corpus is small, the baselines are faster, suggesting a room for \system's improvement in more aggressive caching policy. As the size of corpus increases, \system relatively outperforms more and more across all synthetic datasets. The performance at the high end is summarized in previous section \cref{sec:experiment:e2e}. In terms of storage size, \system generally allocates more space comparing to both SQLite and Lucene but all follow the same trend in a logarithmic scale. In the worst setting, \system uses $2.85\times$ more storage than Lucene.

\input{figures/exp_scalability_2}

\subsection{Tiny IoU Sketch Structure}
\label{sec:apdx:results:tiny_structure}

Apart from \cref{sec:experiment:trafeoff-iou}, we also use \texttt{Cranfield} to explore a restrictive range of $B \in \{1000, 1500, \dots, 3000 \}$ and a excessively wide one for $L \in \{1, 2, 4, 6, 8, \dots, 16 \}$. Recall that $L = 1$ corresponds to the na\"ive hash table index structure. \cref{fig:tradeoff_small} shows four aspects of measurement. The false positive averages (\cref{fig:tradeoff_small:false-positives}) confirm our formulation: for a fixed $B$, there exists some $L^{*}$ that minimizes the error. As $B$ increases, we see the false positive averages decrease across all $L$. Such high numbers of false positives result in a higher search latency (\cref{fig:tradeoff_small:search-latency}) observed near the two ends of dotted line. Although the search latencies across different $B$ ($B \geq 1500 $) are similar due to stochasticity, higher settings of $B$ generally results in a faster search. It should come as no surprise that spending more resource (memory in this case) results in a better performance.

The number of layers $L$ has the most impact on the storage usage because it roughly determines how many bins each posting would belong to. In the worst case where no common postings fall into the same bin, the storage usage is a linear function of $L$; however, \cref{fig:tradeoff_small:storage} reveals a sublinear relationship, especially in small $B$. This is due to a high chance of word's hash collision, inducing more intersection of postings. Developing a hash function that promotes such intersection to save space while controlling false positives is a promising future direction. Moreover, $L$ also approximately linearly affects term lookup latencies as presented in \cref{fig:tradeoff_small:lookup-latency}. Thanks to concurrent network communication, lookup latency is substantially smaller than a multiple of $L$; for example, lookup latency for $L = 16$ is much less than $16\times$ of that for $L = 1$.

\input{figures/exp_iou_tradeoff_small}

\subsection{Tighter Accuracy Requirement}
\label{sec:apdx:results:accuracy}

We also study the relationship between expected false positive parameter $F_0$ on optimal number of layers, and ultimately, search and term index lookup latencies. We select $F_0 \in \{1.0, 0.01, 0.0001\}$ using $B = 10^5$ bins. Then, we observe the optimal numbers of layers $L^{*}$, search latency, and term index lookup latency. Despite differences in magnitude from these setting, the resulting optimal number of layers $L^{*}$ increases only slightly (\cref{fig:exp:acc_tradeoff:layers}). This is consistent with but even stronger (in terms of exponentiation base) than the earlier upper bound analysis that the expected number of false positives is exponentially decreasing at $O(2^{-L})$. As we have seen earlier, a small difference in number of layers with fixed number of bins reflects in a slight difference in both search and term index lookup latencies. Overall, \cref{fig:exp:acc_tradeoff:search} displays slight increases in latencies when accuracy parameter increases. Those observations that do not follow the trend do so because of variability in cloud storage.

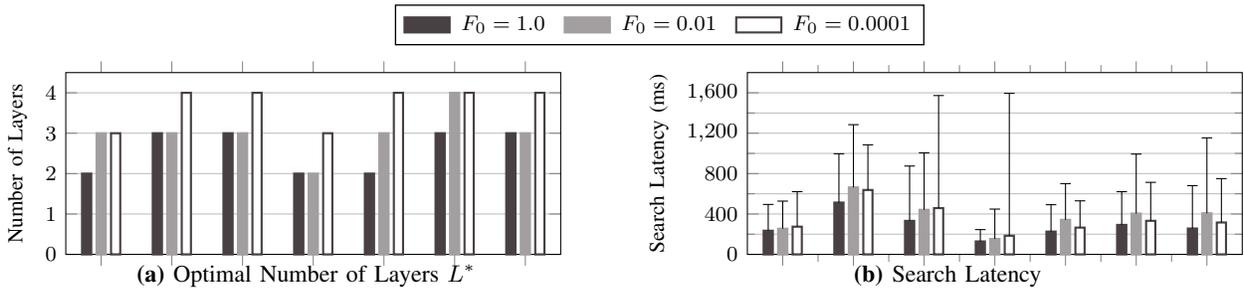
\begin{figure*}[t]
  \centering

\begin{subfigure}[b]{0.45\linewidth}
    \begin{tikzpicture}
    \begin{axis}[
        width=1.0\linewidth,
        height=40mm,
        ybar,
        bar width=0.015\linewidth,
        xmin=0.5,
        xmax=7.5,
        xtick={1, 2, ..., 7},
        xticklabels={},
        ylabel near ticks,
        ylabel style={align=center},
        ymin=0,
        ymax=4.5,
        ytick={0, 1, ..., 4},
        xticklabel style={yshift=2mm},
        xlabel style={yshift=0mm},
        clip=false,
        every axis/.append style={font=\footnotesize},
        minor grid style=lightgray,
        ymajorgrids,
        legend style={
            at={(1.2,1.15)},anchor=south,column sep=2pt,
            draw=black,fill=white,line width=.5pt,
            /tikz/every even column/.append style={column sep=5pt}
        },
        legend cell align={left},
        legend columns=3,
        area legend,
        ylabel=Number of Layers]
    
    \addplot[thick,draw=vintageblack,fill=vintageblack,]
    table[x=x,y=y] {
    x c y    y+
    1	3	2	2558700865
    2	3	3	38867173082
    3	3	3	23842342365
    4	3	2	1727265
    5	3	2	2319845896
    6	3	3	28245668027
    7	3	3	6326309018
    };
    
    \addplot[thick,draw=vintageblack!50,fill=vintageblack!50,]
    table[x=x,y=y] {
    x c y    y+
    1	2	3	3836296334
    2	2	3	38867173082
    3	2	3	23842342365
    4	2	2	1727265
    5	2	3	2748876282
    6	2	4	30572148088
    7	2	3	6326309018
    };
    
    \addplot[thick,draw=vintageblack,fill=white,]
    table[x=x,y=y] {
    x c y    y+
    1	1	3	3836296334
    2	1	4	51819162833
    3	1	4	29606748023
    4	1	3	1998269
    5	1	4	3177900683
    6	1	4	30572148089
    7	1	4	7184161021
    };
    
    \addlegendentry{$F_0 = 1.0$}
    \addlegendentry{$F_0 = 0.01$}
    \addlegendentry{$F_0 = 0.0001$}

    \end{axis}
    \end{tikzpicture}
    
    \vspace{-3mm}
    \caption{Optimal Number of Layers $L^{*}$}
    \label{fig:exp:acc_tradeoff:layers}
\end{subfigure}
~
\begin{subfigure}[b]{0.45\linewidth} 
    \begin{tikzpicture}
    \begin{axis}[
        width=1.0\linewidth,
        height=40mm,
        ybar,
        bar width=0.015\linewidth,
        xmin=0.5,
        xmax=7.5,
        xtick={1, 2, ..., 7},
        xticklabels={},
        ylabel near ticks,
        ylabel style={align=center},
        ymin=0,
        ymax=1800,
        ytick={0, 400, ..., 1600},
        xticklabel style={yshift=2mm},
        xlabel style={yshift=0mm},
        legend style={
            at={(0.46,1.2)},anchor=south,column sep=2pt,
            draw=black,fill=white,line width=.5pt,
            /tikz/every even column/.append style={column sep=5pt}
        },
        legend cell align={left},
        legend columns=3,
        area legend,
        clip=false,
        every axis/.append style={font=\footnotesize},
        minor tick num=1,
        minor grid style=lightgray,
        ymajorgrids,
        yminorgrids,
        ylabel=Search Latency (ms)]
    
    \addplot[thick,draw=vintageblack,fill=vintageblack,]
    plot [error bars/.cd, y dir=plus,y explicit,error bar style={color=black}]
    table[x=x,y=y,y error expr=\thisrowno{3}-\thisrowno{2}] {
    x c y    y+
    1	3	234.2804641	494.399881
    2	3	513.1455453	996.321521
    3	3	331.6545002	874.795758
    4	3	128.1804687	245.284867
    5	3	225.6528376	491.898939
    6	3	291.8882793	620.578197
    7	3	256.1637277	680.811308
    };
    
    \addplot[thick,draw=vintageblack!50,fill=vintageblack!50,]
    plot [error bars/.cd, y dir=plus,y explicit,error bar style={color=black}]
    table[x=x,y=y,y error expr=\thisrowno{3}-\thisrowno{2}] {
    x c y    y+
    1	2	254.2812522	527.467323
    2	2	663.4064984	1284.404663
    3	2	441.5480138	1006.050336
    4	2	152.6476241	448.425081
    5	2	342.173046	700.901926
    6	2	404.3887058	993.233446
    7	2	407.2462167	1153.824425
    };
    
    \addplot[thick,draw=vintageblack,fill=white,]
    plot [error bars/.cd, y dir=plus,y explicit,error bar style={color=black}]
    table[x=x,y=y,y error expr=\thisrowno{3}-\thisrowno{2}] {
    x c y    y+
    1	1	274.5159123	621.050755
    2	1	637.6762959	1084.87594
    3	1	458.2496258	1573.012455
    4	1	184.7343866	1594.48116
    5	1	265.0517705	530.12317
    6	1	332.4198357	713.327716
    7	1	315.8083193	750.085184
    };
    
    
      
    \end{axis}
    \end{tikzpicture}
    
    \vspace{-3mm}
    \caption{Search Latency}
    \label{fig:exp:acc_tradeoff:search}
\end{subfigure}

  \caption{Latencies of \system with different accuracy constraint configurations. Solid bars show average latencies while the upper error bars show 99th percentiles of measured latencies.}
  
  
  \label{fig:acc_tradeoff}
  
  \vspace{-4mm}
\end{figure*}

\end{document}